\documentclass{article}

\usepackage{fullpage}
\usepackage{subfigure}
\usepackage{html}
\usepackage{amsmath,amsfonts,amssymb,amsthm}
\usepackage{graphicx}
\usepackage{grffile}
\usepackage{color}
\usepackage[noend]{algorithmic}
\usepackage{algorithm}
\usepackage{cite}
\usepackage{xspace}
\usepackage{times}
\usepackage{paralist}
\usepackage{hyperref}

\newtheorem{fact}{Fact}[section]
\newtheorem{lemma}[fact]{Lemma}
\newtheorem{theorem}[fact]{Theorem}

\newtheorem{corollary}[fact]{Corollary}

\DeclareMathOperator*{\expv}{Exp}
\DeclareMathOperator*{\poisv}{Pois}
\DeclareMathOperator*{\poly}{poly}
\newcommand{\TV}{{\mathsf{TV}}}
\newcommand{\lscore}{\operatorname{score}}
\newcommand{\argmax}{\operatornamewithlimits{argmax}}
\newcommand{\R}{\mathbb{R}}
\newcommand{\kldiv}[2]{{D( #1 \, \| \, #2)}}

\newcommand{\edgl}{{y}}
\newcommand{\ninf}{{\rho}}
\newcommand{\sinf}[2]{{{#1} \langle {#2} \rangle}}
\newcommand{\vctr}[1]{{\mathbf{#1}}}
\newcommand{\tdel}{{\mathring{t}}}
\newcommand{\vtdel}{{\vctr{\tdel}}}
\newcommand{\expect}{{\mathbb{E}}}
\newcommand{\NetInf}{{\sc NetInf}\xspace}

\begin{document}

\title{Trace Complexity of Network Inference\footnote{A preliminary version of this paper appeared in \cite{ACKP13}.}}
\author{
Bruno Abrahao\thanks{Supported by AFOSR grant FA9550-09-1-0100, by a Cornell University Graduate School Travel Fellowship, and by a Google Award granted to Alessandro Panconesi.}\\
Department of Computer Science\\
Cornell University\\
Ithaca, NY, 14850, USA\\
{\normalsize \tt abrahao@cs.cornell.edu}
\and
Flavio Chierichetti\thanks{Supported by a Google Faculty Research Awards and by the MULTIPLEX project (EU-FET-317532).}\\
Dipartimento di Informatica\\
Sapienza University\\
Rome, Italy\\
{\normalsize \tt flavio@di.uniroma1.it}
\and
Robert Kleinberg\thanks{Supported by AFOSR grant FA9550-09-1-0100, a Microsoft Research New Faculty Fellowship, and a Google Research Grant.}\\
Department of Computer Science\\
Cornell University\\
Ithaca, NY, 14850, USA\\
{\normalsize \tt rdk@cs.cornell.edu}
\and
Alessandro Panconesi\thanks{Supported by a Google Faculty Research Awards and by the MULTIPLEX project (EU-FET-317532).}\\\
Dipartimento di Informatica\\
Sapienza University\\
Rome, Italy\\
{\normalsize \tt ale@di.uniroma1.it}
}

\maketitle




\begin{abstract}

The network inference problem consists of reconstructing the edge set of a network given traces representing the chronology of infection times as epidemics spread through the network. This problem is a paradigmatic representative of prediction tasks in machine learning that require deducing a latent structure from observed patterns of activity in a network, which often require an unrealistically large number of resources (e.g., amount of available data, or computational time). A fundamental question is to understand which properties we can predict with a reasonable degree of accuracy with the available resources, and which we cannot.
We define the \emph{trace complexity} as the number of distinct traces required to achieve high fidelity in reconstructing the topology of the unobserved network or, more generally, some of its properties. We give algorithms that are competitive with, while being simpler and more efficient than, existing network inference approaches.
Moreover, we prove that our algorithms are nearly optimal, by proving an information-theoretic lower bound on the
number of traces that an optimal inference algorithm requires for performing this task in the general case.
Given these strong lower bounds, we turn our attention to special cases, such as trees and bounded-degree graphs, and to property recovery tasks, such as reconstructing the degree distribution without inferring the network. We show that these problems require a much smaller (and more realistic) number of traces, making them potentially solvable in practice.



\end{abstract}

\newpage
\section{Introduction}

Many technological, social, and biological phenomena are naturally modeled as the propagation of a contagion through a network. For instance, in the blogosphere, ``memes'' spread through an underlying social network of bloggers~\cite{Adar05Tracking}, and, in biology, a virus spreads over a population through a network of contacts~\cite{Bailey75Mathematical}. In many such cases, an observer may not directly probe the underlying network structure, but may have access to the sequence of times at which the nodes are infected. Given one or more such records, or \emph{traces}, and a probabilistic model of the epidemic process, we can hope to deduce the underlying graph structure or at least estimate some of its properties. This is the \emph{network inference} problem, which researchers have studied extensively in recent years~\cite{Adar05Tracking, Gomez-Rodriguez2010Inferring, Gomez-Rodriguez2011Uncovering, Du2012Learning,Netrapalli12Learning}.

In this paper we focus on the number of traces that network inference tasks require, which we define as the \emph{trace complexity} of the problem. Our work provides inference algorithms with rigorous upper bounds on their trace complexity, along with information-theoretic lower bounds. We consider network inference tasks under a diffusion model presented in \cite{Gomez-Rodriguez2010Inferring}, whose suitability for representing real-world cascading phenomena in networks is supported by empirical evidence. In short, the model consists of a random cascade process that starts at a single node of a network, and each edge $\{u,v\}$ independently propagates the epidemic, once $u$ is infected, with probability $p$ after a random \emph{incubation time}.

{\bf Overview of results.} In the first part of this paper, we focus on determining the number of traces that are necessary and/or sufficient to perfectly recover the edge set of the whole graph with high probability. We present algorithms and (almost) matching lower bounds for exact inference by showing that in the worst case, $\Omega\left( \frac{n\Delta}{\log^2 \Delta}\right)$ traces are necessary and $O(n \Delta \log n)$ traces are sufficient, where $n$ is the number of nodes in the network and $\Delta$ is its maximum degree. In the second part, we consider a natural line of investigation, given the preceding strong lower bounds, where we ask whether exact inference is possible using a smaller number of traces for special classes of networks that frequently arise in the analysis of social and information networks. Accordingly, we present improved algorithms and trace complexity bounds for two such cases. We give a very simple and natural algorithm for exact inferences of trees that uses only $O(\log n)$ traces.\footnote{All inference results in this paper hold with high probability.}  To further pursue this point, we give an algorithm that exactly reconstructs graphs of degree bounded by $\Delta$ using only $O(\poly(\Delta) \log n)$ traces, under the assumption that epidemics always spread throughout the whole graph. Finally, given that recovering the topology of a hidden network in the worst case requires an impractical number of traces, a natural question is whether some non-trivial property of the network can be accurately determined using a moderate number of traces. Accordingly, we present a highly efficient algorithm that, using vastly fewer traces than are necessary for reconstructing the entire edge set, reconstructs the degree distribution of the network with high fidelity by using $O(n)$ traces.

{\bf The information contained in a trace.} Our asymptotic results also provide some insight into the usefulness of information contained in a trace. Notice that the first two nodes of a trace unambiguously reveal one edge --- the one that connects them. As we keep scanning a trace the signal becomes more and more blurred: the third node could be a neighbor of the first or of the second node, or both. The fourth node could be the neighbor of any nonempty subset of the first three nodes, and so on. The main technical challenge in our context is whether we can extract any useful information from the \emph{tail} of a trace, i.e., the suffix consisting of all nodes from the second to the last. As it turns out, our lower bounds show that, for perfect inference on general connected graphs, the answer is ``no'': we show that the {\em First-Edge  algorithm}, which just returns the edges corresponding to the first two nodes in each trace and ignores the rest, is essentially optimal. This limitation precludes optimal algorithms with practical trace complexity\footnote{On the other hand, the use of short traces may not be only a theoretical limitation, given the real world traces that we observe in modern social networks. For example, Bakshy et al.~\cite{Bakshy11Everyone} report that most cascades in Twitter (\url{twitter.com}) are short, involving one or two hops.}. This  result motivates further exploration of trace complexity for special-case graphs. 
Accordingly, for trees and bounded degree graphs, we illustrate how the tail of traces can be extremely useful for network inference tasks.

Our aforementioned algorithms for special-case graphs make use of maximum likelihood estimation (MLE) but in different ways. Previous approaches, with which we compare our results, have also employed MLE for network inference. For instance, \NetInf~\cite{Gomez-Rodriguez2010Inferring} is an algorithm that attempts to reconstruct the network from a set of independent traces by exploring a submodular property of its MLE formulation. Another example, and closest to ours, is the work by Netrapalli and Sangahvi~\cite{Netrapalli12Learning}, whose results include qualitatively similar bounds on trace complexity in a quite different epidemic model. 

Turning our attention back to our algorithms, our tree reconstruction algorithm performs global likelihood maximization over the entire graph, like the \NetInf algorithm~\cite{Gomez-Rodriguez2010Inferring}, whereas our bounded-degree reconstruction algorithm, like the algorithm in~\cite{Netrapalli12Learning}, performs MLE at each individual vertex. Our algorithms and analysis techniques, however, differ markedly from those of~\cite{Gomez-Rodriguez2010Inferring} and \cite{Netrapalli12Learning}, and may be of independent interest.

In the literature on this rapidly expanding topic, researchers have validated their findings using small or stylized graphs and a relatively large number of traces. In this work, we aim to provide, in the same spirit as \cite{Netrapalli12Learning}, a formal and rigorous understanding of the potentialities and limitations of algorithms that aim to solve the network inference problem.

This paper is organized as follows. Section~\ref{sec:relwork} presents an overview of previous approaches to network learning. Section~\ref{sec:model} presents the cascade model we consider throughout the paper. Section~\ref{sec:firstedge} deals with the \emph{head of the trace}: it presents the First-Edge algorithm for network inference, shows that it is essentially optimal in the worst case, and shows how the first edges' timestamps can be used to guess the degree distribution of the network. Section~\ref{sec:tail}, instead, deals with the \emph{tail of the trace}: it presents efficient algorithms for perfect reconstruction of the topology of trees and of bounded degree networks.  Section~\ref{sec:exp} presents an experimental analysis that compares ours and existing results through the lens of trace complexity. Section~\ref{sec:conc} offers our conclusions. The proofs missing from the main body of the paper can be found in  Appendix~\ref{app:proofs}.

\section{Related Work}
\label{sec:relwork}

Network inference has been a highly active area of investigation in data mining and machine learning~\cite{Adar05Tracking, Gomez-Rodriguez2010Inferring, Gomez-Rodriguez2011Uncovering, Du2012Learning, Netrapalli12Learning}. It is usually assumed that an event  initially activates one or more nodes in a network,  triggering a cascading process, e.g., bloggers acquire a piece of information that interests other bloggers~\cite{Gruhl04Information}, a group of people are the first infected by a contagious virus~\cite{Bailey75Mathematical}, or a small group of consumers are the early adopters of a new piece of technology that subsequently becomes popular~\cite{Rogers03Diffusion}. In general, the process spreads like an epidemic over a network (i.e., the network formed by blog readers, the friendship network, the coworkers network). Researchers derive observations from each cascade in the form of {\em traces} ---  the identities of the people that are activated in the process and the timestamps of their activation. However, while we do see traces, we do not directly observe the network over which the cascade spreads. The network inference problem consists of recovering the underlying network using the epidemic data.

In this paper we study the cascade model that Gomez-Rodrigues et al.~\cite{Gomez-Rodriguez2010Inferring} introduced, which consists of a variation of the independent cascade model~\cite{Kempe2003Maximizing}. Gomez-Rodrigues et al.  propose \NetInf, a maximum likelihood algorithm, for network reconstruction. 
Their method is evaluated under the exponential and power-law distributed incubation times. In our work, we restrict our analysis to the case where the incubation times are exponentially distributed as this makes for a rich arena of study. 

Gomez-Rodrigues et al. have further generalized the model to include different transmission rates for different edges and a broader collection of waiting times distributions~\cite{Gomez-Rodriguez2011Uncovering, myers10OntheConvexity}. Later on, Du et al.~\cite{Du2012Learning} proposed a kernel-based method that is able to recover the network without prior assumptions on the waiting time distributions. These methods have significantly higher computational costs than \NetInf, and, therefore, than ours. Nevertheless, experiments on real and synthetic data show a marked improvement in accuracy, in addition to gains in flexibility. Using a more combinatorial approach, Gripon and Rabbat~\cite{Gripon13Reconstructing} consider the problem of reconstructing a graph from traces defined as sets of unordered nodes, in which the nodes that appear in the same trace are connected by a path containing exactly the nodes in the trace. In this work, traces of size three are considered, and the authors identify necessary and sufficient conditions to reconstruct graphs in this setting.


The performance of network inference algorithms is dependent on the amount of information available for the reconstruction, i.e., the number and length of traces. The dependency on the number of traces have been illustrated in \cite{Du2012Learning}, \cite{Gomez-Rodriguez2011Uncovering}, and \cite{Gomez-Rodriguez2010Inferring} by plotting the performance of the algorithms against the number of available traces. Nevertheless, we find little research on a rigorous analysis of this dependency, with the exception of one paper~\cite{Netrapalli12Learning} that we now discuss. 


Similarly to our work, Netrapalli and Sangahvi~\cite{Netrapalli12Learning} present quantitative bounds on trace complexity in a quite different epidemic model. The model studied in \cite{Netrapalli12Learning} is another variation of the independent cascade model. It differs from the model we study in a number of key aspects, which make that model a simplification of the model we consider here. For instance, (i) \cite{Netrapalli12Learning} assumes a cascading process over discrete time steps, while we assume continuous time (which has been shown to be a realistic model of several real-world processes~\cite{Gomez-Rodriguez2010Inferring}),  (ii) the complexity analyzed in \cite{Netrapalli12Learning} applies to a model where nodes are active for a single time step --- once a node is infected, it has a single time step to infect its neighbors, after which it becomes permanently inactive.
The model we consider does not bound the time that a node can wait before infecting a neighbor. Finally, (iii) 
\cite{Netrapalli12Learning} rely crucially on the ``correlation decay'' assumption, which implies -- for instance --- that each node
can be infected during the course of the epidemics by {\em less} than 1 neighbor in expectation.
The simplifications in the model presented by \cite{Netrapalli12Learning}  make it less realistic --- and,  also, make the inference task significantly easier than the one we consider here.


We believe that our analysis introduces a rigorous foundation to assess the performance of existing and new algorithms for network inference. In addition, to the best of our knowledge, our paper is the first to study how different parts of the trace can be useful for different network inference tasks. Also, it is the first to study the trace complexity of special case graphs, such as bounded degree graphs, and for reconstructing non-trivial properties of the network (without reconstructing the network itself), such as the node degree distribution.

\section{Cascade Model}
\label{sec:model}

The cascade model we consider
is defined as follows. 
It starts with one activated node, henceforth called the {\em source} of the epidemic, which is considered to be activated, without loss of generality, at time $t=0$. 

As soon as a node $u$ gets activated, for each neighbor $v_i$, 
 $u$ flips an independent coin: with probability $p$ it will start a countdown on the edge $\{u,v_i\}$. The length of the countdown will be a random variable distributed according to $\expv(\lambda)$ (exponential\footnote{\cite{Gomez-Rodriguez2010Inferring, Gomez-Rodriguez2011Uncovering, Du2012Learning} consider other random timer distributions; we will mainly be interested in exponential variables as this setting is already rich enough to make for an interesting and extensive analysis.} with parameter $\lambda$). When the countdown reaches $0$, that edge is {\em traversed} --- that is, that epidemic reaches $v_i$ via $u$.

The ``trace'' produced by the model will be a sequence of tuples (node $v$, $t(v)$) where $t(v)$ is the first time at which the epidemics reaches $v$.

In \cite{Gomez-Rodriguez2010Inferring}, the source of the epidemics is chosen uniformly at random from the nodes of the network.
In general, though, the source can be chosen arbitrarily\footnote{Choosing sources in a realistic way is an open problem --- the data that could offer a solution to this problem seems to be extremely scarce at this time.}.

The cascade process considered here admits a number of equivalent  descriptions. The following happens to be quite handy: independently for each edge of $G$, remove the edge with probability $1-p$ and otherwise assign a random edge length sampled from $\expv(\lambda)$.  Run Dijkstra's single-source shortest path algorithm on the subgraph formed by the edges that remain, using source $s$ and the sampled edge lengths. Output vertices in the order they are discovered, accompanied by a timestamp representing the shortest
path length. 


\section{The Head of a Trace}
\label{sec:firstedge}

In this section we will deal with the head of a trace --- that is, with the edge connecting the first and the second nodes of a trace. We show that, for general graphs, that edge is the only useful information that can be extracted from traces. Moreover, and perhaps surprisingly, this information is enough to achieve close-to-optimal trace complexity, i.e., no network inference algorithm can achieve better performance than a simple algorithm that only extracts the head of the trace and ignores the rest. We analyze this algorithm in the next section.

\subsection{The First-Edge Algorithm}
\label{sec:firstedge-algo}

The First-Edge algorithm is simple to state. For each trace in the input, it extracts the edge connecting the first two nodes, and adds this edge the guessed edge set, ignoring the rest of the trace. This procedure is not only optimal in trace complexity, but, as it turns out, it is also computationally efficient.

We start by showing that First-Edge is able to reconstruct the full graph with maximum degree $\Delta$ using $\Theta(n \Delta \log n)$ traces, under the cascade model we consider. 

\begin{theorem}
Suppose that the source $s \in V$ is chosen uniformly at random. Let $G=(V,E)$ be a graph with maximum degree $\Delta \le n - 1$.
With $\Theta\left(\frac{n \Delta}p \log n\right)$ traces, First-Edge correctly returns the graph $G$ with probability at least $1 - \frac1{\poly(n)}$.
\end{theorem}
\begin{proof}
Let $e = \{u,v\}$ be any edge in $E$. The probability that a trace starts with $u$, and continues with $v$ can be lower bounded by $\frac p{n\Delta}$, that is, by the product of the probabilities that $u$ is selected as the source, that the edge $\{u,v\}$ is not removed from the graph, and that $v$ is the first neighbor of $u$ that gets infected. Therefore, if we run  $c \frac{n\Delta}p \ln n$ traces, the probability that none of them starts with the ordered couple of neighboring nodes $u,v$ is at most:
$$\left(1 - \frac{p}{n \Delta}\right)^{\frac{n\Delta}p c \ln n} \le \exp(-c \ln n) = n^{-c}.$$


\noindent Therefore, the assertion is proved for any constant $c > 2$.
\end{proof}



We notice that a more careful analysis leads to a proof that $$\Theta\left(\left(\Delta + p^{-1}\right) n \log n\right)$$  traces are enough to reconstruct the whole graph with high probability. To prove this stronger assertion, it is sufficient to show the probability that a specific edge will be the first one to be traversed is at least $\frac2n \cdot \left(1-e^{-1}\right) \cdot \min\left(\Delta^{-1},p\right)$.
In fact one can even show that, for each $d \le \Delta$, if the First-Edge algorithm has access to $O\left(\left(d +p^{-1}\right) n \log n\right)$ traces, then it will recover all the edges having at least one endpoint of degree less than or equal $d$. 
As we will see in our experimental section, this allows us to reconstruct a large fraction of the edges using a number of traces that is significantly smaller than the maximum degree times the number of nodes.

\smallskip

Finally, we note that the above proof also entails that First-Edge performs as stated for any waiting time distribution (that is, not just for the exponential one). In fact, the only property that we need for the above bounds to hold, is that 
the first node, and the first neighbor of the first node, are chosen independently and uniformly at random by the process.

\subsection{Lower Bounds}\label{sec:lb}

In this section we discuss a number of lower bounds for network inference. 

\smallskip

We start by observing that if the source node is chosen adversarially --- and, say if the graph is disconnected --- no algorithm can reconstruct the graph (traces are trapped in one connected component and, therefore, do not contain any information about the rest of the graph.) 
Moreover, even if the graph is forced to be connected, by choosing $p = \frac12$ (that is, edges are traversed with probability $\frac12$) an algorithm will require at least $2^{\Omega(n)}$ traces even if the graph is known to be a path. Indeed, if we select one endpoint as the source, it will take $2^{\Omega(n)}$ trials for a trace to reach the other end of the path, since at each node, the trace flips an unbiased coin and dies out with probability $\frac12$.

This is the reason why we need the assumption that the epidemic selects $s\in V$ uniformly at random --- we recall that this is also an assumption in \cite{Gomez-Rodriguez2010Inferring}. Whenever possible, we will consider more realistic assumptions, and determine how this changes the trace complexity of the reconstruction problem.

\smallskip

We now turn our attention to our main lower bound result. Namely, even if traces never die (that is, if $p = 1$), and assuming that the source is chosen uniformly at random, we need $\tilde{\Omega}(n \Delta)$ traces to reconstruct the graph.

\smallskip

First, let $G_0$ be the clique on the node set $V = \{1,\ldots,n\}$, and let $G_1$ be the clique on $V$ minus the edge $\{1,2\}$.

Suppose that Nature selects the unknown graph uniformly at random in the set $\{G_0,G_1\}$. We will show that with $o\left(\frac{n^2}{\log^{2} n}\right)$ traces, the probability that we are able to guess the unknown graph is at most $\frac12 + o(1)$ --- that is, flipping a coin is close to being the best one can do for guessing the existence of the edge $\{1,2\}$.

\smallskip

Before embarking on this task, though, we show that this result directly entails that $o(n \cdot \frac{\Delta}{\log^2 \Delta})$ traces are not enough for reconstruction even if the graph has maximum degree $\Delta$, for each $1 \le \Delta \le n-1$. Indeed, let the graph $G'_0$ be composed of a clique on $\Delta+1$ nodes, and of $n-\Delta-1$ disconnected nodes. Let $G'_1$ be composed of a clique on $\Delta+1$ nodes, minus an edge, and of $n-\Delta-1$ disconnected nodes. Then, due to our yet-unproven lower bound, we need at least $\Omega\left(\frac{\Delta^2}{\log^2 \Delta}\right)$ traces to start in the large connected component for the reconstruction to succeed. The probability that a trace starts in the large connected component is  $O\left(\frac{\Delta}n\right)$. Hence, we need at least $\Omega\left(n\cdot \frac{\Delta}{\log^2 \Delta}\right)$ traces.

\medskip

We now highlight the main ideas that we used to prove the main lower bound, by stating the intermediate lemmas that lead to  it. The proofs of these Lemmas can be found in Appendix~\ref{app:proofs}.

\smallskip

The first lemma states  that the random ordering of nodes produced by a trace in $G_0$ is uniform at random, and that the random ordering produced by a trace in $G_1$ is ``very close'' to being uniform at random. Intuitively, this entails that one needs many traces to be able to infer the unknown graph by using the orderings given by the traces.

\begin{lemma}\label{lem:lb:comb}
Let $\pi$ be the random ordering of nodes produced by the random process on $G_0$, and $\pi'$ be the random ordering of nodes produced by the random process on $G_1$.
Then,\begin{compactenum}
\item $\pi$ is a uniform at random permutation over $[n]$;
\item for each $1 \le a < b \le n$, the permutation $\pi'$ conditioned on the vertices $1,2$ appearing (in an arbitrary order) in the positions $a,b$, is uniform at random in that set;
\item moreover, the probability $p_{a,b}$ that $\pi'$ has the vertices $1,2$ appearing (in an arbitrary order) in the positions $a < b$ is equal to
$p_{a,b} = \frac{1+ d(a,b)}{\binom{n}2}$, with\begin{compactitem}
\item  $d(a,b) = -1$ if $a=1,b=2$; otherwise $d(a,b) > -1$;
\item moreover $d_{a,b} = O\left(\frac{\ln n}{n}\right) - O\left(\frac1b\right)$.
\end{compactitem}
\item Finally, $\sum_{a = 1}^{n-1} \sum_{b=a+1}^n d(a,b) = 0$.
\end{compactenum}
\end{lemma}

The preceding Lemma can be used to prove Lemma~\ref{lem:lb:perm}: if one is forced not to used timestamps, $o\left(\frac{n^2}{\log^2n}\right)$ traces are not enough to guess the unknown graph with probability more than $\frac12 + o(1)$. 
\begin{lemma}\label{lem:lb:perm}
Let $\mathcal{P}$ the sequence of the $\ell$ orderings of nodes given by $\ell$ traces, with $\ell = o\left(\frac{n^{2}}{\ln^2 n}\right)$.

The probability that the likelihood of $\mathcal{P}$ is higher in the  graph $G_0$ is equal to $\frac12 \pm o(1)$, regardless of the unknown graph.
\end{lemma}

The next Lemma, which also needs Lemma~\ref{lem:lb:comb}, takes care of the waiting times in the timestamps. Specifically, it shows that -- under a conditioning having high probability -- the probability that the sequence of waiting times of the traces has higher likelihood in $G_0$ than in $G_1$ is $\frac12 \pm o(1)$, regardless of the unknown graph.

\begin{lemma}\label{lem:lb:waitingtimes}
Let $\alpha$ satisfy $\alpha = o(1)$, and $\alpha = \omega\left(\frac{\log n}n\right)$.
Also, let $\ell_i$ be the number of traces that have exactly one of the nodes in $\{1,2\}$ the first $i$ informed nodes.

Let $\mathcal{W}$ be the random waiting times of the traces. Then, if we condition on  $\ell_i = \Theta\left(\alpha \cdot i \cdot (n-i)\right)$  for each $i=1,\ldots, n$ (independently of the actual node permutations), the probability that the likelihood of $\mathcal{W}$ is higher in the graph $G_0$ is equal to $\frac12 \pm o(1)$, regardless of the unknown graph.
\end{lemma}

Finally, the following corollary follows directly from Lemma~\ref{lem:lb:perm} and Lemma~\ref{lem:lb:waitingtimes}, and by a trivial application of the Chernoff Bound.

\begin{corollary}\label{cor:lb}
If Nature chooses between $G_0$ and $G_1$ uniformly at random, and one has access to $o\left(\frac{n^2}{\log^2 n}\right)$ traces, then no algorithm can correctly guess the graph with probability more than $\frac12 + o(1)$.\end{corollary}

\smallskip

As already noted, the lower bound of Corollary~\ref{cor:lb} can be easily transformed in a $\Omega\left(n \cdot \frac{\Delta}{\log^2 \Delta}\right)$ lower bound, for any $ \Delta \le n -1$.

\subsection{Reconstructing the Degree Distribution}
\label{sec:degree}

In this section we study the problem of recovering the degree distribution of a hidden network and show that this can be done with $\Omega(n)$ traces while achieving high accuracy, using, again, only the first edge of a trace.

The degree distribution of a network is a characteristic structural property of networks, which influences their dynamics, function, and evolution~\cite{Newman03Thestructure}. Accordingly, many networks, including the Internet and the world wide web exhibit distinct degree distributions~\cite{Faloutsos99OnPower}. Thus, recovering this property allows us to make inferences about the behavior of processes that take place in these networks, without knowledge of their actual link structure.

Let $\ell$ traces starting from the same node $v$ be given. For trace $i$, let $t_i$ be the differences between the time of exposure of $v$, and the the time of exposure of the second node in the trace.

Recall that in the cascade model, the waiting times are distributed according to an exponential random variable with a known parameter $\lambda$. If we have $\ell$ traces starting at a node $v$, we aim to estimate the degree of $v$ the time gaps $t_1,\ldots, t_{\ell}$ between the first node and the second node of each trace.

If $v$ has degree $d$ in the graph, then $t_i$ ($1 \le i \le \ell$) will be distributed as an exponential random variable with parameter $d \lambda$~\cite{Durrett11Probability}.  Furthermore, the sum $T$ of the $t_i$'s, $T = \sum_{i=1}^{\ell} t_i$, is distributed as an Erlang random variable with parameters $(\ell, d\lambda)$~\cite{Durrett11Probability}.

\medskip

In general, if $X$ is an Erlang variable with parameters $(n, \lambda)$, and $Y$ is a Poisson variable with parameter $z \cdot \lambda$, we  have that
$\Pr\left[X < z\right] = \Pr\left[Y \ge n\right]$.
Then, by using the tail bound for the Poisson distribution~\cite{bl98, km06}, we have that the probability that $T$ is at most $(1+\epsilon) \cdot \frac{\ell}{d \lambda}$ is
$$\Pr\left[\poisv\left( (1+\epsilon) \cdot \ell\right) \ge \ell\right] \ge 1 - e^{-\Theta\left(\epsilon^2 \ell\right)}.$$

Similarly, the probability that $T$ is at least $(1-\epsilon) \cdot
\frac{\ell}{d \lambda}$ is 
$$1 - \Pr\left[\poisv((1-\epsilon) \cdot \ell) \ge \ell\right] \ge 1 - e^{-\Theta\left(\epsilon^2 \ell\right)}.$$

We then have:
$$\Pr\left[\left|T - \frac{\ell}{d \lambda}\right| \le \epsilon \cdot \frac{\ell}{d \lambda}\right] \ge 1 - e^{-\Theta\left(\epsilon^2 \ell\right)}.$$

Let our degree inference algorithm return $\hat{d} = \frac{\ell}{T \lambda}$ as the degree of $v$. Also, let $d$ be the actual degree of $v$. We have:
$$\Pr\left[\left|\hat{d} - d\right| \le \epsilon d\right] \ge 1 - e^{-\Theta\left(\epsilon^2 \ell\right)}.$$
We have then proved the following theorem:
\begin{theorem}
Provided that $\Omega\left(\frac{\ln \delta^{-1}}{\epsilon^2}\right)$ traces start from $v$, the degree algorithm returns a $1 \pm \epsilon$ multiplicative approximation to the degree of $v$ with probability at least $1 - \delta$.
\end{theorem}

\section{The Tail of the Trace}
\label{sec:tail}

A na\"{i}ve interpretation 
of the lower bound for perfect reconstruction, Corollary~\ref{cor:lb},
would conclude that the information in the ``tail'' of the trace 
--- the list of
nodes infected after the first two nodes, 
and their timestamps ---
is of negligible use in achieving the task of 
perfect reconstruction. In this section we will see
that the opposite conclusion holds for important
classes of graphs. 
We specialize to two such classes, trees and
bounded-degree graphs, in both cases designing
algorithms that rely heavily on information in the
tails of traces to achieve perfect reconstruction
with trace complexity $O(\log n)$,
an exponential improvement from the worst-case
lower bound in Corollary~\ref{cor:lb}. The 
algorithms are quite different: for trees
we essentially perform maximum likelihood estimation
(MLE) of the entire edge set all at once,
while for bounded-degree graphs we run MLE separately
for each vertex to attempt to find its set of neighbors,
then we combine those sets while resolving
inconsistencies. 

In Section~\ref{sec:exp}
we provide one more example of an algorithm,
which we denote by First-Edge$+$, 
that makes use of information in the tail
of the trace. Unlike the algorithms in this
section, we do not know of a theoretical
performance guarantee for First-Edge$+$
so we have instead analyzed it experimentally.

It is natural to compare the algorithms in this
section with the \NetInf algorithm~\cite{Gomez-Rodriguez2010Inferring},
since both are based on MLE. While \NetInf is a general-purpose
algorithm, and the algorithms developed here are limited to
special classes of graphs, we believe our approach offers
several advantages.
First, and most importantly,
we offer provable trace complexity guarantees: $\Omega(\log n)$
complete traces suffice for perfect reconstruction of a tree
with high probability, and $\Omega(\poly(\Delta) \log n)$ traces
suffice for perfect reconstruction of a graph with maximum degree
$\Delta$. Previous work has not provided rigorous guarantees on
the number of traces required to ensure that algorithms achieve
specified reconstruction tasks. Second, our tree reconstruction
algorithm is simple 
(an easy preprocessing step followed by computing a minimum
spanning tree) and has worst-case running time $O(n^2 \ell)$,
where $n$ is the number of nodes and  $\ell = \Omega(\log n)$ 
is the number of traces, which compares favorably with the
running time of \NetInf.

\subsection{Reconstructing Trees}
\label{sec:trees}

In this section we consider the special case in which the underlying
graph $G$ is a tree, and we provide a simple algorithm that requires
$\Omega(\log n)$ complete traces and succeeds in perfect reconstruction
with high probability. Intuitively, reconstructing trees is much
simpler than reconstructing general graphs for the following reason.
As noted in~\cite{Gomez-Rodriguez2010Inferring}, the probability
that an arbitrary 
graph $G$ generates trace $T$ is a sum, over all spanning trees
$F$ of $G$, of the probability that $T$ was generated by an epidemic
propagating along the edges of $F$. When $G$ itself is a tree, this
sum degenerates to a single term and this greatly simplifies the
process of doing maximum likelihood estimation.
In practical applications of the network inference problem,
it is unlikely that the latent network will be a tree; 
nevertheless we believe the results in this section are of
theoretical interest and that they may provide a roadmap for
analyzing the trace complexity of other algorithms
based on maximum likelihood estimation.

\begin{algorithm}
\begin{algorithmic}[1]
\REQUIRE A collection $T_1,\ldots,T_\ell$ of complete traces generated by repeatedly
running the infection process with $p=1$ on a fixed tree.\\
Let $t_i(v)$ denote the infection time of node $v$ in trace $T_i$. 
\ENSURE An estimate, $\hat{G}$, of the tree.
\FORALL{pairs of nodes $u,v$}
\STATE Let $c(u,v)$ be the median of the set $\{|t_i(u)-t_i(v)|\}_{i=1}^{\ell}$.
\IF{$\exists$ a node $p$ and a pair of traces $T_i, T_j$
such that $t_i(p) < t_i(u) < t_i(v)$ and $t_j(p) < t_j(v) < t_j(u)$}
\STATE Set $c(u,v)=\infty$.
\ENDIF
\ENDFOR
\STATE Output $\hat{G} =$ minimum spanning tree
with respect to cost matrix  $c(u,v)$.
\end{algorithmic}
\caption{The tree reconstruction algorithm.}\label{alg:tree}
\end{algorithm}

The tree reconstruction algorithm is very simple. It defines a
cost for each edge $\{u,v\}$ as shown in Figure~\ref{alg:tree},
and then it outputs the minimum spanning tree with respect to
those edge costs. The most time-consuming step is the test in step 3,
which checks whether there is a node $p$ whose infection time
precedes the infection times of both $u$ and $v$ in two
distinct traces $T_i,T_j$ such that the infection times of $u$ and $v$
are oppositely ordered in $T_i$ and $T_j$. 
(If so, then $G$ contains a path from
$p$ to $u$ that does not include $v$, and a path from $p$ to $v$
that does not include $u$, and consequently $\{u,v\}$ cannot be 
an edge of the tree $G$. This justifies setting 
$c(u,v)=\infty$ in step 4.) 
To save time, one can use lazy evaluation 
to avoid performing
this test for every pair $u,v$. The lazy version of the algorithm
computes edge costs 
$c(u,v)$ as in step 3 and then proceeds straight to the minimum
spanning tree computation, using Kruskal's algorithm. 
Any time Kruskal's algorithm decides to insert an
edge $\{u,v\}$ into the tree, we instead perform the test
in step 3 and delete edge $\{u,v\}$ from the graph if it violates
the test.

The analysis of the algorithm is based on the following outline:
first, we show that if $\{u,v\}$ is any edge of $G$, then
$c(u,v) < \lambda^{-1}$ with high probability (Lemma~\ref{lem:tree-edge}). 
Second, we show that if 
$\{u,v\}$ is any edge not in $G$, then $c(u,v) > \lambda^{-1}$ with
high probability (Lemma~\ref{lem:non-tree-edge}). The edge pruning
in steps 3 and 4 of the algorithm is vital for attaining the
latter high-probability guarantee. 
When both of these high-probability 
events occur, it is trivial to see that the minimum
spanning tree coincides with $G$.

\begin{lemma} \label{lem:tree-edge}
If $\{u,v\}$ is an edge of the tree $G$, then 
Algorithm~\ref{alg:tree} sets $c(u,v) < \lambda^{-1}$
with probability at least $1-{c_1}^{\lambda}$, 
for some absolute constant $c_1 < 1$.
\end{lemma}
\begin{proof}
First, note that the algorithm never sets $c(u,v)=\infty$.
This is because if one were to delete edge $\{u,v\}$ from $G$, it 
would disconnect the graph into two connected components 
$G_u, G_v$, containing $u$ and $v$, respectively.
The infection process cannot 
spread from $G_u$ to $G_v$ or vice-versa without 
traversing edge $\{u,v\}$. Consequently, for every
node $p \in G_u$, the infection time $t_i(u)$ occurs
strictly between $t_i(p)$ and $t_i(v)$ in all traces.
Similarly, if $p \in G_v$ then the infection time 
$t_i(v)$ occurs strictly between $t_i(p)$ and $t_i(u)$
in all traces.

Therefore, the value of $c(u,v)$ is equal to the median
of $|t_i(u)-t_i(v)|$ over all the traces $T_1,\ldots,T_\ell$.
In any execution of the infection process, if the first
endpoint of edge $\{u,v\}$ becomes infected at time $t$,
then the opposite endpoint receives a timestamp $t+X$
where $X \sim \expv(\lambda)$. Consequently the random
variable $|t_i(u)-t_i(v)|$ is an independent sample
from $\expv(\lambda)$ in each trace.
The probability 
that any one of these samples is greater than $\lambda^{-1}$
is $1/e$, so the probability that their median exceeds
$\lambda^{-1}$ is equal to the probability of observing
at least $\ell/2$ heads in $\ell$ tosses of a coin with
bias $1/e$. By Chernoff's bound~\cite{MR95}, this is less than
$(\sqrt{2} e^{1/e})^{-\ell}$. 
\end{proof}

The remaining step in analyzing the tree reconstruction
algorithm is to prove that $c(u,v) > \lambda^{-1}$
with high probability when $\{u,v\}$ is not an edge 
of the tree $G$. 
\begin{lemma} \label{lem:non-tree-edge}
If $\{u,v\}$ is not an edge of $G$, then
Algorithm~\ref{alg:tree} sets $c(u,v) > \lambda^{-1}$
with probability at least $1 - c_2 \cdot c_3^{\ell}$ for some
absolute constants $c_2 < \infty$ and $c_3 < 1$.
\end{lemma}
\begin{proof}
$G$ is a tree, so for any two nodes $u,v$, there is a 
unique path $P(u,v)$ in $G$ that starts at $u$ and 
ends at $v$. Furthermore, for every $s \in G$, there
is a unique node $z(s) \in P(u,v)$ such that the
paths $P(s,u)$ and $P(s,v)$ are identical up until
they reach $z(s)$, and they are vertex-disjoint afterward.
When the infection process starts at $s$ and
spreads throughout $G$, it always holds that
$t(z(s)) \leq \min \{t(u),t(v)\}$. Conditional
on the value of $t(z(s))$, the infection times
of vertices on the paths $P(z(s),u)$ and 
$P(z(s),v)$ constitute two independent Poisson
processes each with rate $\lambda$. 
Let $n_u(s)$ and $n_v(s)$ denote the number of edges in the 
paths $P(z(s),u)$ and $P(z(s),v)$, respectively.
The infection
times $t(u), t(v)$ occur at the $n_u(s)^{\mathrm{th}}$
and $n_v(s)^{\mathrm{th}}$ arrival times, respectively,
in the two independent Poisson processes.

Let $s_1,\ldots,s_\ell$ denote the sources of traces
$T_1,\ldots,T_\ell$. We distinguish two cases. First,
suppose at least $\frac{\ell}{10}$ of the traces satisfy 
$n_u(s_i)=n_v(s_i)$. In any of these traces, the events
$t_i(u) < t_i(v)$ and $t_i(v) < t_i(u)$ both have 
probability $1/2$, by symmetry. 
Hence, with probability
at least $1-2 \, \cdot \, 2^{-\ell/10}$, there exist traces 
$T_i, T_j$ such that $z(s_i), z(s_j)$ are both
equal to the midpoint of the path $P(u,v)$, but
$t_i(u) < t_i(v)$ whereas $t_j(v) < t_j(u)$. 
If this high-probability event happens, the 
condition in step 3 of the algorithm will be
satisfied with $p = z(s_i) = z(s_j)$ and the 
cost $c(u,v)$ will be set to $\infty$.

The remaining case is that at least $\frac{9\ell}{10}$
of the traces satisfy $n_u(s_i) \neq n_v(s_i)$. 
In this
case, we reason about the distribution of $|t_i(u)-t_i(v)|$
as follows. 
Let $q$ denote the number of uninfected nodes on
path $P$ at the time $t$ when an element of $\{u,v\}$ is first
infected. Conditional on the value of $t$,
the remaining infection times of the nodes
on path $P$ are the arrival times in a Poisson process
of rate $\lambda$. The conditional probability that
$|t_i(u)-t_i(v)| > \lambda^{-1}$, given $q$, is
therefore equal to the probability that a $\poisv(1)$
random variable is less than $q$. This conditional
probability is equal to $1/e$ when $q=1$ and 
is at least $2/e$ when $q > 1$. (The value of $q$
is always at least 1, because at time $t$ exactly
one element of $\{u,v\}$ is infected and the other 
is not yet infected.) 

When $n_u(s_i) \neq n_v(s_i)$,
we claim that $\Pr(q>1)$ is at least 1/2. To see why,
assume without loss of generality that 
$n_u(s_i) < n_v(s_i)$ and let $x$ be the 
node on path $P(u,v)$ 
such that $x \neq u$ but $u$ and $x$ are equidistant
from $z(s_i)$. (In other words, the paths
$P(z(s_i),x)$ and $P(z(s_i),u)$ have the same
number of edges.) By symmetry, the events
$t_i(u) < t_i(x)$ and $t_i(x) < t_i(u)$ both
have probability 1/2. Conditional on the event
$t_i(u) < t_i(x)$, we have $q > 1$ because 
$x, v$ are two distinct nodes that are
uninfected at time $t_i(u)$. Consequently,
$\Pr(q>1) \geq 1/2$ as claimed.

Now let us combine the conclusions of the preceding two
paragraphs. For notational convenience, we use
$t_i^{uv}$ as shorthand for $|t_i(u)-t_i(v)|$. 
When $n_u(s_i) \neq n_v(s_i)$
we have derived:
\begin{align*}
\Pr(t_i^{uv} > \lambda^{-1}) &=
\Pr(t_i^{uv} > \lambda^{-1} \,\mid\, q=1) \Pr(q=1)
 + 
\Pr(t_i^{uv} > \lambda^{-1} \,\mid\, q>1) \Pr(q>1) \\
&\geq
\tfrac12 \left( \tfrac{1}{e} \right) +
\tfrac12 \left( \tfrac{2}{e} \right) = \tfrac{1.5}{e}.
\end{align*}
When $n_u(s_i) = n_v(s_i)$ we have derived:
\begin{align*}
\Pr(t_i^{uv} > \lambda^{-1}) &\geq
\Pr(t_i^{uv} > \lambda^{-1} \,\mid\, q=1) = \tfrac{1}{e}.
\end{align*}
Recall that $c(u,v)$ is the median of $t_i^{uv}$
for $i=1,\ldots,\ell$. The probability that this 
median is less than $\lambda^{-1}$ is 
bounded above by the probability of observing
fewer than $\ell/2$ heads when tossing $\ell/10$
coins with bias $\frac{1}{e}$ and $9 \ell/10$ coins
with bias $\frac{1.5}{e}$. The expected number
of heads in such an experiment is $\frac{0.1 + (0.9)(1.5)}{e} = 
\frac{1.45}{e} > \frac{8}{15}$. Once again applying 
Chernoff's bound (to the random variable that counts
the number of \emph{tails}) the probability that at least 
$\ell/2$ tails are observed is bounded above by 
$\left( \frac{14}{15} e^{1/15} \right)^{\ell/2} < (0.999)^{\ell}.$
\end{proof}

Combining Lemmas~\ref{lem:tree-edge} and~\ref{lem:non-tree-edge},
and using the union bound, we find that with probability
at least $1 - (n-1)c_1^{\ell} - \binom{n-1}{2} c_2 c_3^{\ell}$,
the set of pairs $(u,v)$ such
that $c(u,v) < \lambda^{-1}$ coincides with the set of edges
of the tree $G$. Whenever the $n-1$ cheapest edges in a graph
form a spanning tree, it is always the minimum
spanning tree of the graph.
Thus, we have proven the following theorem.

\begin{theorem} \label{thm:tree}
If $G$ is a tree, then Algorithm~\ref{alg:tree}
perfectly reconstructs $G$ with probability
at least $1 - (n-1)c_1^{\ell} - \binom{n-1}{2} c_2 c_3^{\ell}$,
for some absolute constants $c_1, c_3 < 1$ and $c_2 < \infty$.
This probability can be made greater than $1 - 1/n^c$,
for any specified $c>0$, by using $\ell \geq c_4 \cdot c  \cdot \log n$
traces, where $c_4 < \infty$ is an absolute constant.
\end{theorem}

\subsection{Bounded-Degree Graphs}
\label{sec:bdd-degree}

In this section, we show that $O(\poly(\Delta) \log n)$ complete
traces suffice for perfect reconstruction (with high probability)
when the graph $G$ has maximum degree $\Delta$. In fact, our
proof shows a somewhat stronger result: it shows that for any
pair of nodes $u,v$, there is an 
algorithm that predicts whether $\{u,v\}$ is an edge of $G$ 
with failure probability at most $1-1/n^c$, for any 
specified constant $c>0$, and the algorithm requires only
$\Omega(\poly(\Delta) \log n)$ independent partial
traces in which $u$ and $v$ are both infected.
However, for simplicity we will assume complete traces
throughout this section. 

\begin{algorithm}
\begin{algorithmic}[1]
\REQUIRE 
An infection rate parameter, $\lambda$.\\
A set of vertices, $V$.\\
An upper bound, $\Delta$, on the degrees of vertices.\\
A collection $T_1,\ldots,T_\ell$ of
complete traces generated by repeatedly
running the infection process on a fixed graph $G$ with vertex
set $V$ and maximum degree $\Delta$.\\
Let $t_i(v)$ denote the infection time of node $v$ in trace $T_i$.
\ENSURE An estimate, $\hat{G}$, of $G$.
\FORALL{nodes $u$}
\FORALL{sets $S \subseteq V \setminus \{u\}$ of at most $\Delta$ vertices}
\FORALL{traces $T_i$} 
\STATE Let $S_i^u = \{v \in S \mid t_i(v) < t_i(u)\}$.
\IF{$S_i^u = \emptyset$}
\STATE Let $\lscore_i(S,u) = 0$ if $u$ is the source of $T_i$,
otherwise $\lscore_i(S,u) = -\infty$.
\ELSE
\STATE $\lscore_i(S,u) = \log |S_i^u| - \lambda \sum_{v \in S_i^u} [t_i(u)-t_i(v)]$.
\ENDIF
\ENDFOR
\STATE Let $\lscore(S,u) = \ell^{-1} \cdot \sum_i \lscore_i(S,u)$.
\ENDFOR
\STATE Let $R(u) = \argmax \{\lscore(S,u)\}.$
\ENDFOR
\FORALL{ordered pairs of vertices $u,v$}
\IF{$t_i(v) < t_i(u)$ in at least
$\ell/3$ traces and $v \in R(u)$}
\STATE Insert edge $\{u,v\}$ into $\hat{G}$.
\ENDIF
\ENDFOR
\STATE Output $\hat{G}$.
\end{algorithmic}
\caption{Bounded-degree reconstruction algorithm.}\label{alg:bdd-degree}
\end{algorithm}

The basic intuition behind our algorithm can be summarized
as follows. To determine if $\{u,v\}$ is an edge of $G$, we 
try to reconstruct the entire set of neighbors of $u$ and
then test if $v$ belongs to this set. We use the following
insight to test whether a candidate set $S$ is equal to the
set $N(u)$ of all neighbors of $u$. Any such set defines a 
``forecasting model'' that specifies a probability distribution for the 
infection time $t(u)$. To test the validity of the forecast
we use a strictly proper scoring rule~\cite{Gneiting07Strictly}, 
specifically the logarithmic scoring rule, which is defined formally 
in the paragraph following Equation~\eqref{eq:kldiv-continuous}.
Let us say that a set $S$ differs
significantly from the set of neighbors of $u$ (henceforth
denoted $N(u)$) if the symmetric difference $S \oplus N(u)$ contains a vertex
that is infected before $u$ with constant probability.
We prove that the expected score assigned to $N(u)$ by the
logarithmic scoring rule is at least $\Omega(\Delta^{-4})$ greater
than the score assigned to any set that differs significantly from
$N(u)$. Averaging over $\Omega(\Delta^4 \log \Delta \log n)$
trials is then sufficient to ensure that all sets differing 
significantly from $N(u)$ receive strictly smaller
average scores.

The scoring rule algorithm thus succeeds (with high probability) 
in reconstructing a 
set $R(u)$ whose difference from $N(u)$ is insignificant, meaning
that the elements of $R(u) \oplus N(u)$ are usually infected
after $u$. To test if edge $\{u,v\}$ belongs to $G$, we can now
use the following procedure: if the event $t(v) < t(u)$ occurs
in a constant fraction of the traces containing both $u$ and $v$, then
we predict that edge $\{u,v\}$ is present if $v \in R(u)$;
this prediction must be correct with high probability, as otherwise
the element $v \in R(u) \oplus N(u)$ would constitute a significant
difference. Symmetrically, if $t(u) < t(v)$ occurs in a constant
fraction of the traces containing both $u$ and $v$, then 
we predict that edge $\{u,v\}$ is present if 
$u \in R(v)$.


\smallskip

{\bf KL-divergence.}
For distributions $p,q$ on $\R$ having
density functions $f$ and $g$, respectively,
their KL-divergence is defined
 by
 \begin{equation} \label{eq:kldiv-continuous}
 \kldiv{p}{q} = \int f(x) \log \left( \tfrac{f(x)}{g(x)} \right) \, dx.
 \end{equation}
One interpretation of the KL-divergence is that it is 
the expected difference between $\log(f(x))$ and $\log(g(x))$
when $x$ is randomly sampled using distribution $p$.
If one thinks of $p$ and $q$ as two forecasts
of the distribution of $x$, and one samples $x$ using
$p$ and applies the \emph{logarithmic scoring rule},
which outputs a score equal to the log-density of the
forecast distribution at the sampled point, then
$\kldiv{p}{q}$ is the difference in the expected
scores of the correct and the incorrect forecast.
A useful lower bound on this difference is supplied by
Pinsker's Inequality:
 \begin{equation} \label{eq:pinsker}
 \kldiv{p}{q} \geq 2 \, \|p-q\|^2_{\TV},
 \end{equation}
 where $\| \cdot \|_{\TV}$ denotes the total
 variation distance.
In particular, the fact that $\kldiv{p}{q} > 0$ when
$p \neq q$ 
means that the true distribution, $p$, is the unique
distribution that attains the maximum expected score,
a property that is summarized by stating that the 
logarithmic scoring rule is \emph{strictly proper}.

\smallskip

{\bf Quasi-timestamps and conditional distributions}
From now on in this section, we assume $\lambda=1$. This
assumption is without loss of generality, since the algorithm's
behavior in unchanged if we modify its input by setting 
$\lambda=1$ and multiplying the timestamps in all traces
by $\lambda$; after modifying the input in this way, the
input distribution is the same as if the traces had originally
been sampled
using the infection process with parameter $\lambda=1$.

Our analysis of Algorithm~\ref{alg:bdd-degree} hinges on
understanding the conditional distribution of the infection
time $t(u)$, given the infection times of its neighbors. 
Directly analyzing this conditional distribution is surprisingly
tricky, however. The reason is that $u$ itself may infect some
of its neighbors, so conditioning on the event that a neighbor
of $u$ was infected at time $t_0$ influences the probability
density of $t(u)$ in a straightforward way at times $t>t_0$ but
in a much less straightforward way at times $t<t_0$. We can
avoid this ``backward conditioning'' by applying the following artifice.

Recall the description of the infection
process in terms of Dijkstra's algorithm
in Section~\ref{sec:model}: edges sample
i.i.d.\ edge lengths and the timestamps $t(v)$
are equal to the distance labels assigned by
Dijkstra's algorithm when computing single-source
shortest paths from source $s$.
Now consider the sample space defined by the tuple 
of independent random edge lengths $\edgl(v,w)$.
For any vertices
$u \neq v$, define a random variable $\tdel(v)$ to be
the distance label assigned to $v$ when we \emph{delete}
$u$ and its incident edges from $G$ to obtain a 
subgraph $G-u$, and then we run Dijkstra's
algorithm on this subgraph. One can think of
$\tdel(v)$ as the time when $v$ would have been
infected if $u$ did not exist. We will call 
$\tdel(v)$ the \emph{quasi-timestamp of $v$}
(with respect to $u$).
If $N(u)=\{v_1,\ldots,v_k\}$ is the set of neighbors
of $u$, and if we sample a trace originating at a source
$s \neq u$, then the executions of 
Dijkstra's algorithm in $G$ and 
$G-u$ will coincide until the step in
which $u$ is discovered and is assigned the 
distance label
$t(u) = \min_j \{\tdel(v_j) + \edgl(v_j,u)\}.$
From this equation, it is easy to deduce a formula
for the conditional distribution of $t(u)$ given
the $k$-tuple of quasi-timestamps 
$\vtdel = (\tdel(v_j))_{j=1}^k$.
Using the standard notation $z^+$ to denote $\max\{z,0\}$
for any real number $z$, we have
\begin{equation} \label{eq:cond-cdf}
\Pr(t(u) > t \mid \vtdel) = 
\exp \left( - 
\sum_{j=1}^k (t - \tdel(v_j))^+ \right).
\end{equation}
The conditional probability density is easy
to calculate by differentiating the right side
of~\eqref{eq:cond-cdf} with respect to $t$.
For any vertex set $S$ not containing $u$, let
$\sinf{S}{t}$ denote the set of vertices
$v \in S$ such that $\tdel(v) < t$, and let
$\ninf(t,S) = \left| \sinf{S}{t} \right|$.
Then the conditional probability density function
of $t(u)$ satisfies
\begin{align} \label{eq:cond-pdf}
f(t) &= 
\ninf(t,N(u)) \exp \left( -
\sum_{j=1}^k (t - \tdel(v_j))^+ \right) \\
\log f(t) &= 
\log(\ninf(t,N(u))) -
\sum_{v \in N(u)} (t-\tdel(v))^+.
\label{eq:cond-logprob}
\end{align}
It is worth pausing here to note an important
and subtle point. The
information contained in a trace $T$ is 
insufficient to 
determine the vector of quasi-timestamps
$\vtdel$, since quasi-timestamps are 
defined by running
the infection process in the graph $G-u$,
whereas the trace represents the outcome
of running the same process in $G$.
Consequently, our algorithm does not have
sufficient information to evaluate 
$\log f(t)$ at arbitrary values of $t$.
Luckily, the equation 
$$
(t(u)-t(v))^+ = (t(u) - \tdel(v))^+
$$
holds for all $v \neq u$, since $\tdel(v)$
differs from $t(v)$ only when both quantities
are greater than $t(u)$. Thus, our algorithm
has sufficient information to evaluate 
$\log f(t(u))$, and in fact the value
$\lscore_i(S,u)$ defined in Algorithm~\ref{alg:bdd-degree},
coincides with the formula for
$\log f(t(u))$
on the right side of~\eqref{eq:cond-logprob},
when $S = N(u)$ and $\lambda=1$.

\smallskip

{\bf Analysis of the reconstruction algorithm.} The foregoing discussion prompts the following
definitions. Fix a vector of quasi-timestamps 
$\vtdel = (\tdel(v))_{v \neq u}$, and for 
any set of vertices $S$ not containing
$u$, let $p^S$ be the probability distribution
on $\R$ with density function 
\begin{equation} \label{eq:fs}
f^S(t) = \ninf(t,S) 
\exp \left( - \sum_{v \in S} (t - \tdel(v))^+ \right).
\end{equation}
One can think of $p^S$ as the distribution 
of the infection time $t(u)$ that would be predicted by a 
forecaster who knows the values $\tdel(v)$ 
for $v \in S$ and who believes that $S$ is the
set of neighbors of $u$.
Letting $N = N(u)$, each timestamp 
$t_i(u)$
is a random sample from the distribution $p^N$, 
and $\lscore_i(S,u)$ is the result
of applying the logarithmic scoring rule to
the distribution $p^S$ and the random
sample $t(u)$. Therefore
\begin{align} 
\expect[\lscore_i(N,u) - \lscore_i(S,u)] = 
\kldiv{p^N}{p^S} \geq 
2 \| p^N - p^S \|^2_{\TV}.
\label{eq:tv} 
\end{align}
The key to analyzing Algorithm~\ref{alg:bdd-degree}
lies in proving a lower bound on the expected 
total variation distance between $p^N$ and $p^S$.
The following lemma supplies the lower bound.
\begin{lemma} \label{lem:tv-lb}
Fix a vertex $u$, let $N=N(u)$ be its neighbor set,
and fix some $S \subseteq V \setminus \{u\}$ distinct
from $N$.
Letting $\pi(S \oplus N,u)$ denote
the probability that at least one element of 
the set $S \oplus N$ is infected
before $u$, we have
\begin{equation} \label{eq:tv-lb}
\expect \left( \|p^N - p^S \|_{\TV} \right) \geq 
\tfrac{1}{10} \Delta^{-2} \pi(S \oplus N, u).
\end{equation}
\end{lemma}
\begin{proof}
For a fixed vector of quasi-timestamps $(\tdel(v))_{v \neq u}$
we can bound $\|p^N - p^S\|_{\TV}$ from below by the 
following method. Let $v_0$ denote the vertex in 
$S \oplus N$ whose quasi-timestamp $t_0$ is earliest.
Let $b$ be the largest number in the range $0 \leq b \leq \frac{1}{\Delta}$
such that the open interval $I=(t_0,t_0+b)$ does not contain
the quasi-timestamps of any element of $S \cup N$.
The value $|p^N(I) - p^S(I)|$ is a lower bound on 
$\|p^N-p^S\|_{\TV}$. 

One may verify by inspection
that the density
function $f^S(t)$ defined in equation~\eqref{eq:fs}
satisfies the differential equation 
$f^S(t) = \frac{d}{dt} \left( f^S(t) / \ninf(t,S) \right)$
for almost all $t$.
By integrating both sides of the equation we find
that for all $t$,
\begin{equation} \label{eq:1-fs}
1 - F^S(t) = \frac{f^S(t)}{\ninf(t,S)} = 
\exp \left(- \sum_{v \in S} (t-\tdel(v))^+ \right),
\end{equation}
where $F^S$ denotes the cumulative distribution function
of $p^S$. A similar formula holds for $F^N$.
Let $G = 1-F^S(t_0) = 1-F^N(t_0)$, where the latter
equation holds because of our choice of $t_0$.
We have
\begin{align*}
p^S(I) & = G - (1 - F^N(t_0+b)) \\
&= G \cdot (1 - e^{-\ninf(t_0+b,S) \, b}) \\
p^N(I) &= G - (1 - F^S(t_0+b)) \\
&= G \cdot (1 - e^{-\ninf(t_0+b,N) \, b}),
\end{align*}
where the second and fourth lines follow
from the formula~\eqref{eq:1-fs}, using
the fact that none of the quasi-timestamps
$\tdel(v)$ for $v \in S \cup N$ occur in 
the interval $(t_0,t_0+b)$.

Let $\ninf = \ninf(t_0,S) = \ninf(t_0,N)$.
We have
\begin{align} \nonumber
|p^N(I) - p^S(I)| &= G \cdot \left|
e^{-\ninf(t_0+b,N) \, b} - e^{-\ninf(t_0+b,S) \, b} \right| \\
\label{eq:tv-lb-2}
& = G e^{-\ninf b} (1 - e^{-b}),
\end{align}
using the fact that exactly one of $\ninf(t_0+b,S),
\ninf(t_0+b,N)$ is equal to $\ninf$ and the other
is equal to $\ninf+1$. To bound the right side
of~\eqref{eq:tv-lb-2}, we reason as follows.
First, $\ninf = |\sinf{N}{t_0}| \leq |N| \leq \Delta$.
Second, $b \leq \Delta^{-1}$ by construction. Hence
$e^{-\ninf b} \geq e^{-1}$. Also, the inequality 
$1 - e^{-x} \geq (1 - e^{-1})x$ holds for all $x \in [0,1]$,
since the left side is a concave function, the right
side is a linear function, and the two sides agree
at $x=0$ and $x=1$. Thus, $1-e^{-b} \geq (1-e^{-1})b$,
and we have derived
\[
|p^N(I) - p^S(I)| \geq (e^{-1} - e^{-2}) G b.
\]

To complete the proof of the lemma we need to
derive a lower bound on the expectation of the
product $Gb$. First note that $G = 1 - F^N(t_0)$
is the probability $t(u) > t_0$ when $t(u)$ is 
sampled from the distribution $p^N$. Since 
$p^N$ is the conditional distribution of $t(u)$
given $\vtdel$, we can now take the expectation
of both sides of the equation $G = 1-F^N(t_0)$
and conclude that $\expect[G] = \pi(S \oplus N,u)$.
Finally, to place a lower bound on $\expect[b \mid G]$,
we reason as follows. In the infection process on
$G - u$, let $R$ denote the set of vertices in
$S \cup N$ whose quasi-timestamps are strictly
greater than $t_0$. The number of edges joining
$R$ to the rest of $V \setminus \{u\}$ is at most
$\Delta |R| < 2 \Delta$, so the waiting time from
$t_0$ until the next quasi-timestamp of an 
element of $R$ stochastically dominates
the minimum of $2 \Delta^2$ i.i.d.\ $\expv(1)$ random
variables. Thus the conditional distribution 
of $b$ given $G$ stochastically dominates
the minimum of $2 \Delta^2$ i.i.d.\ $\expv(1)$
random variables and the constant $1/\Delta$, so
\[
\expect[b|G] \geq 
\int_0^{1/\Delta} e^{-2 \Delta^2 t} \, dt 
= \tfrac12 \Delta^{-2} \left[ 1 - e^{-2 \Delta} \right] 
\geq \tfrac12 \Delta^{-2} \left[ 1 - e^{-2} \right].
\]
Putting all of these bounds together, we obtain
\[
\expect(\|p^N-p^S\|_{\TV}) \geq \tfrac12 \Delta^{-2} (e^{-1} - e^{-2})
(1 - e^{-2}) \pi(S \oplus N,u),
\]
and the inequality~\eqref{eq:tv-lb} follows
by direct calculation.
\end{proof}

Combining Pinsker's Inequality with Lemma~\ref{lem:tv-lb}
we immediately obtain the following corollary.

\begin{corollary} \label{cor:scrule}
If $N=N(u)$ and $S$ is any set such that 
$\pi(S \oplus N,u) > 1/4$, then for each 
trace $T_i$ the expected value of 
$\lscore_i(N)-\lscore_i(S)$ is
$\Omega(\Delta^{-4})$.
\end{corollary}

Using this corollary, we are ready to prove our main
theorem.
\begin{theorem} \label{thm:bdd-degree}
For any constant $c>0$, the
probability that Algorithm~\ref{alg:bdd-degree}
fails to perfectly reconstruct $G$, when given
$$\ell=\Omega(\Delta^9 \log^2 \Delta \log n)$$
complete traces, is at most $1/n^c$.
\end{theorem}
\begin{proof}
Let us say that a set $S$ \emph{differs significantly}
from $N(u)$ if $\pi(S \oplus N(u),u) > 1/4$.
When $\ell$ is as specified in
the theorem statement, with probability
at least $1 - 1/n^{c+1}$, there is no
vertex $u$ such that the algorithm's
estimate of $u$'s neighbor set, $R(u)$, differs
significantly from $N(u)$.
Indeed, when $S,N$ satisfy the hypotheses
of Corollary~\ref{cor:scrule}, the random 
variables $\lscore_i(N) - \lscore_i(S)$
are i.i.d.\ samples from a distribution
that has expectation $\Omega(\Delta^{-4})$,
is bounded above by $O(\log \Delta)$ with
probability $1 - 1/\poly(\Delta)$, and 
has an exponential tail. Exponential
concentration inequalities for such distributions
imply that for all $\delta>0$, the average of 
$\ell = \Omega(\Delta^{8} \log^2(\Delta) \log(1/\delta))$
i.i.d.\ samples will be non-negative with
probability at least $1-\delta$. Setting 
$\delta = n^{-\Delta-c-2}$ and taking the union
bound over all vertex sets $S$ of cardinality 
$\Delta$ or smaller, we conclude that when
$\ell = \Omega(\Delta^9 \log^2(\Delta) \log n)$,
the algorithm has less than $n^{-c-2}$ probability
of selecting a set $R(u)$ that differs significantly
from $N(u)$. Taking the union bound over all vertices
$u$ we obtain a proof of the claim stated earlier in
this paragraph: with probability $1-1/n^{c+1}$, there
is no $u$ such that $R(u)$ differs significantly
from $N(u)$.

Let us say that an ordered pair of vertices $(u,v)$
violates the \emph{empirical frequency property} if
the empirical
frequency of the event $t_i(v) < t_i(u)$ among the traces
$T_1,\ldots,T_\ell$ differs by more than $\frac{1}{12}$ 
from the probability that $t(v) < t(u)$ in a random trace.
The probability of any given pair $(u,v)$ violating this property
is exponentially small in $\ell$, hence we can assume 
it is
less than $1/n^{c+3}$ by taking the constant inside the $\Omega(\cdot)$
to be sufficiently large. Summing over pairs $(u,v)$, the 
probability that there exists a pair violating the empirical 
frequency property is less than $1/n^{c+1}$ and we henceforth
assume that no such pair exists.

Assuming that no set $R(u)$ differs significantly from $N(u)$
and that no pair $(u,v)$ violates the empirical frequency property,
we now prove that the algorithm's output, $\hat{G}$, is equal to $G$.
If $\{u,v\}$ is an edge of $G$, assume without loss of generality 
that the event $t(v) < t(u)$ has probability at least 1/2. By the
empirical frequency property, at least $\ell/3$ traces satisfy
$t_i(v) < t_i(u)$. Furthermore, $v$ must belong to $R(u)$, since
if it belonged to $R(u) \oplus N(u)$ it would imply that 
$\pi(R(u) \oplus N(u),u) \geq \Pr(t(v) < t(u)) \geq 1/2$,
violating our assumption that $R(u)$ doesn't differ 
significantly from $N(u)$. Therefore $v \in R(u)$ and the
algorithm adds $\{u,v\}$ to $\hat{G}$. Now suppose 
$\{u,v\}$ is an edge of $\hat{G}$, and assume without 
loss of generality that this edge was inserted when 
processing the ordered pair $(u,v)$. Thus, at least $\ell/3$ traces 
satisfy $t_i(v) < t_i(u)$, and $v \in R(u)$. By the 
empirical frequency property, we know that a random trace
satisfies $t(v) < t(u)$ with probability at least $1/4$.
As before, if $v$ belonged to $R(u) \oplus N(u)$ this would
violate our assumption that $R(u)$ does not differ significantly
from $N(u)$. Hence $v \in N(u)$, which means that $\{u,v\}$ is
an edge of $G$ as well.
\end{proof}

\section{Experimental Analysis}
\label{sec:exp}

In the preceding sections we have established trace complexity results for various network inference tasks. In this section, our goal is to assess our predictions on real and synthetic social and information networks whose type, number of nodes, and maximum degree ($\Delta$) we now describe.

\begin{figure*}[ht!] 
    \begin{center} 
    \begin{tabular}{cccc}
    \subfigure[Barabasi-Albert Graph]{\includegraphics[width=5cm]{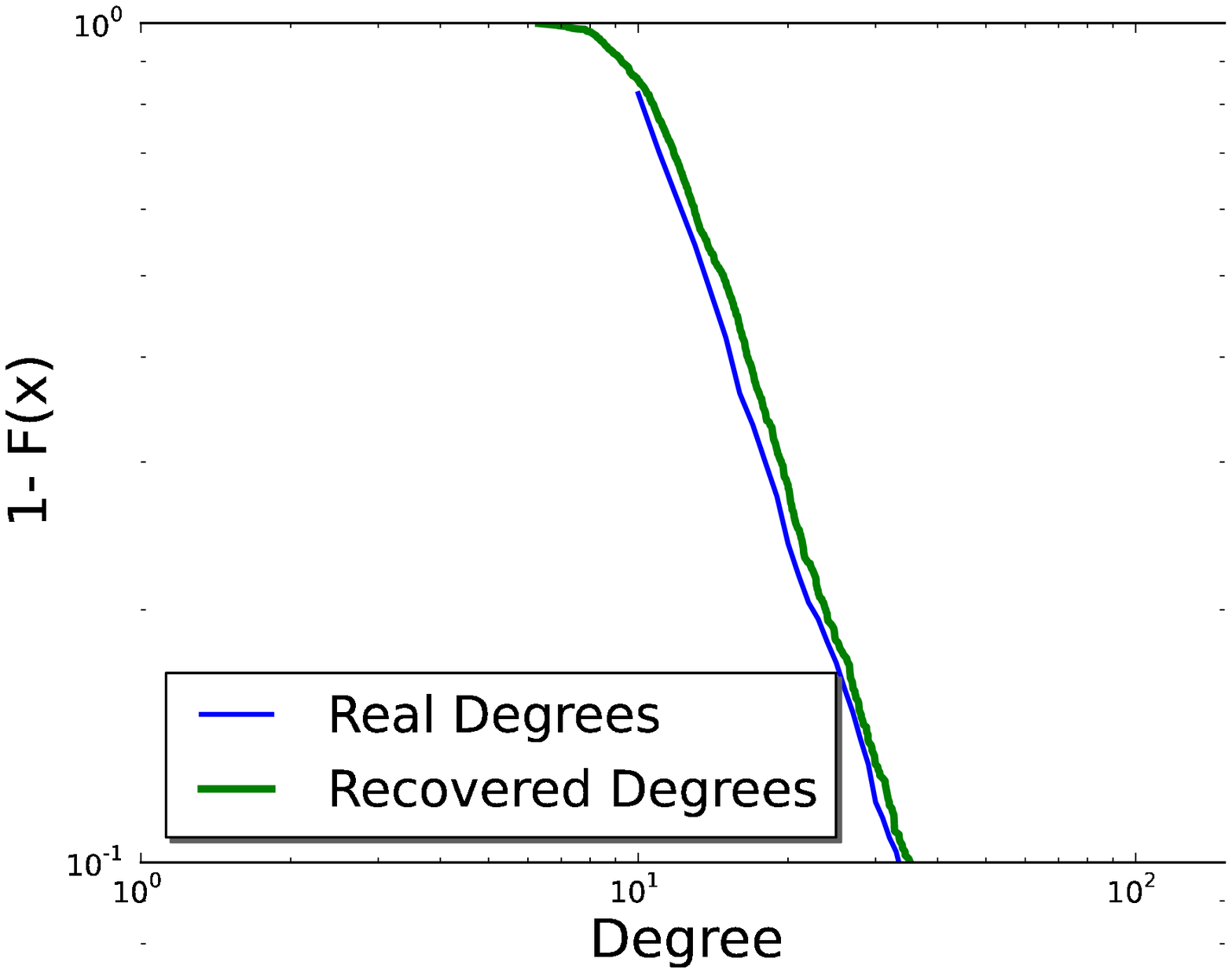}\label{fig:label1}}&
    \subfigure[Facebook-Rice-Graduate]{\includegraphics[width=5cm]{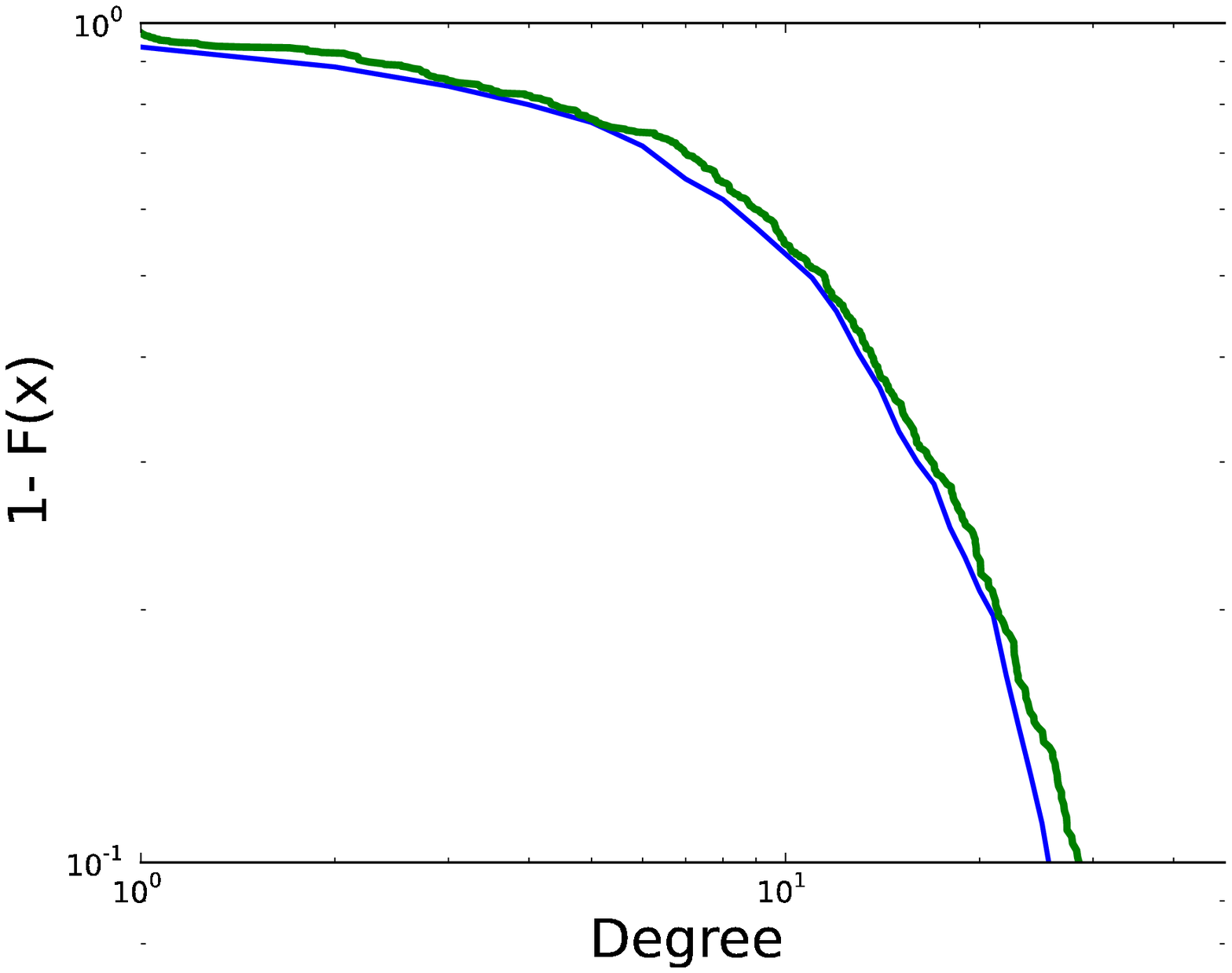}\label{fig:label2}}&
    \subfigure[Facebook-Rice Undergraduate]{\includegraphics[width=5cm]{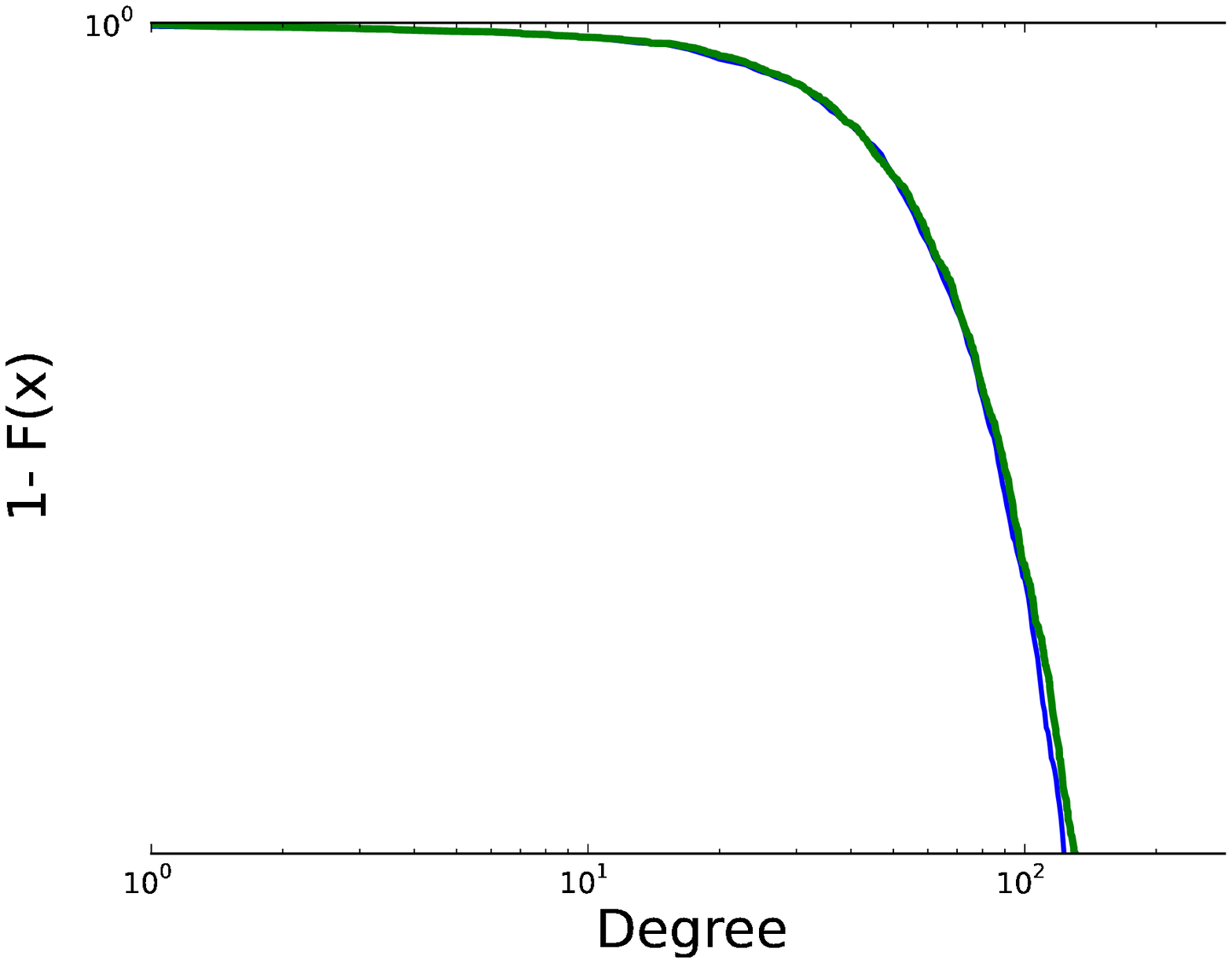}\label{fig:label3}}
    \end{tabular} 
    \end{center} 
    \vspace{-0.5cm}
    \caption{Complementary cumulative density function (CCDF) of degree reconstruction using $\Omega(n)$ traces for (a) a synthetic network with 1,024 nodes
generated using the Barabasi-Albert algorithm, and two real social networks: two subsets of the Facebook network comprising 503 graduate students (a) and 1220 undergraduate students (c), respectively, from Rice University.} 
    \label{fig:reconstructing}
\end{figure*}

We use two real social networks, namely two Facebook subnetworks comprising 503 ($\Delta=48$) graduate and 1220 ($\Delta=287$)  undergraduate students, respectively~\cite{Mislove10You}. We also generate three synthetic networks, each possessing 1024 vertices, whose generative models frequently arise in practice in the analysis of networks.  We generated a \emph{Barabasi-Albert Network}~\cite{Barabasi99Emergence} ($\Delta=174$), which is a preferential attachment model,
a \emph{$G_{(n,p)}$ Network}~\cite{Erdos60OnTheEvolution} ($\Delta=253$) with $p=0.2$, and a \emph{Power-Law Tree}, whose node degree distribution follows a power-law distribution with exponent $3$  ($\Delta=94$).


First, we evaluate the performance of the algorithm to reconstruct the degree distribution of networks without inferring the network itself (Section~\ref{sec:degree}). Figure~\ref{fig:reconstructing} shows the reconstruction of the degree distribution using $\Omega(n)$ traces of the Barabasi-Albert Network and the two Facebook subnetworks. We used $10n$ traces, and the plots show that the CCDF curves for the real degrees and for the reconstructed distribution have almost perfect overlap.

Turning our attention back to network inference, the $\Omega(n \Delta^{1-\epsilon})$ lower-bound established in Section~\ref{sec:model} tells us that the First-Edge algorithm is nearly optimal for perfect network inference in the general case. Thus, we  assess the performance of our algorithms against this limit. The performance of First-Edge is notoriously predictable: if we use $\ell$ traces where $\ell$ is less than the total number of edges in the network, then it returns nearly $\ell$ edges which are all true positives, and it never returns false positives. 

If we allow false positives, we can use heuristics to improve the First-Edge's recall. To this end, we propose the following heuristic that uses the degree distribution reconstruction algorithm (Section~\ref{sec:degree}) in a pre-processing phase, and places an edge in the inferred network provided the edge has probability at least $p$ of being in the graph. We call this heuristic \emph{First-Edge$+$}.

\begin{figure}[ht] 
    \begin{center} 
    \begin{tabular}{cc}
    \subfigure[Barabasi-Albert]{\includegraphics[width=5cm]{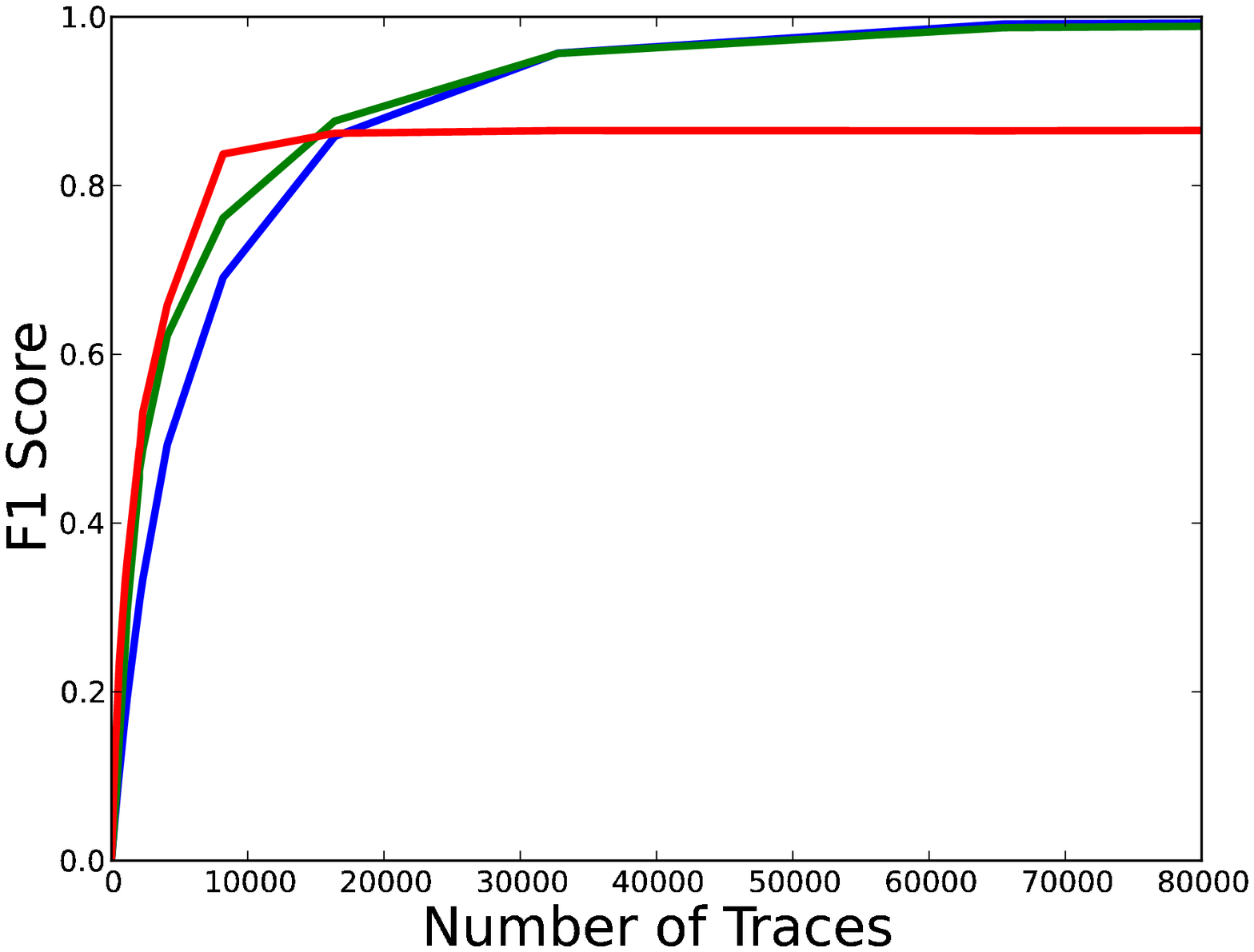}\label{fig:r1}}&
    \subfigure[Facebook-Rice Undergrad]{\includegraphics[width=5cm]{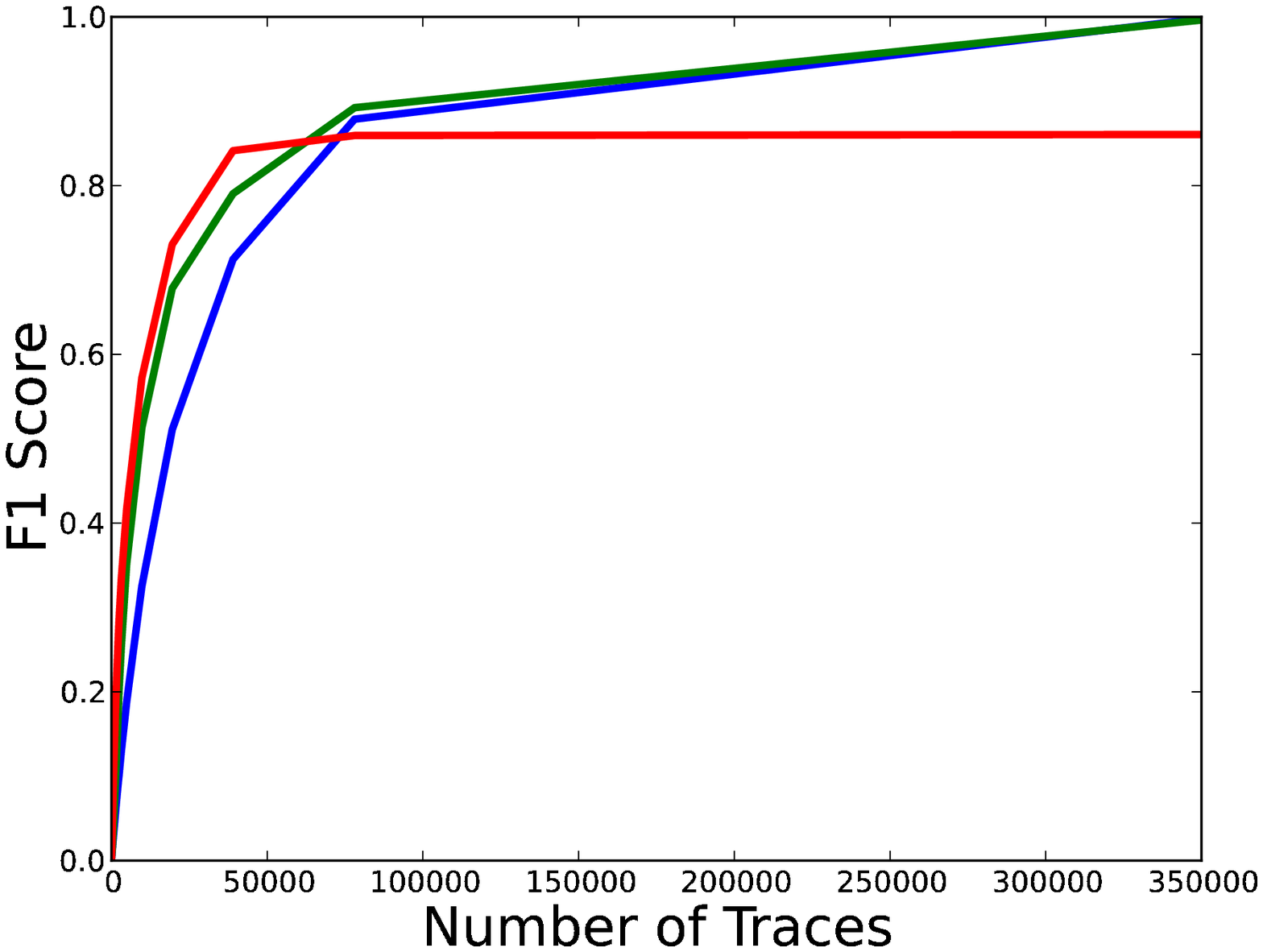}\label{fig:r2}}\\
    \subfigure[Power-Law Tree]{\includegraphics[width=5cm]{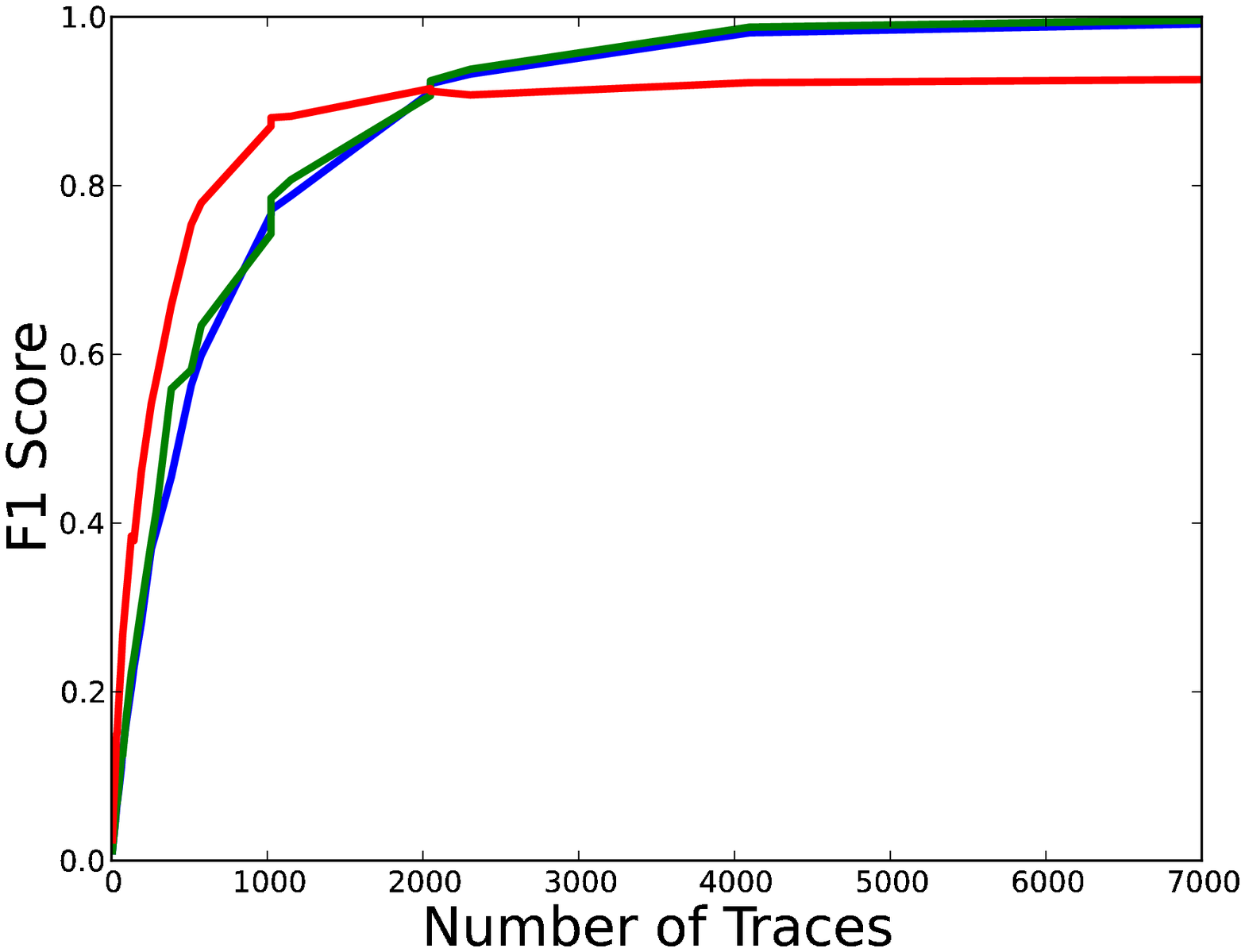}\label{fig:r3}}&
    \subfigure[$G_{n,p}$]{\includegraphics[width=5cm]{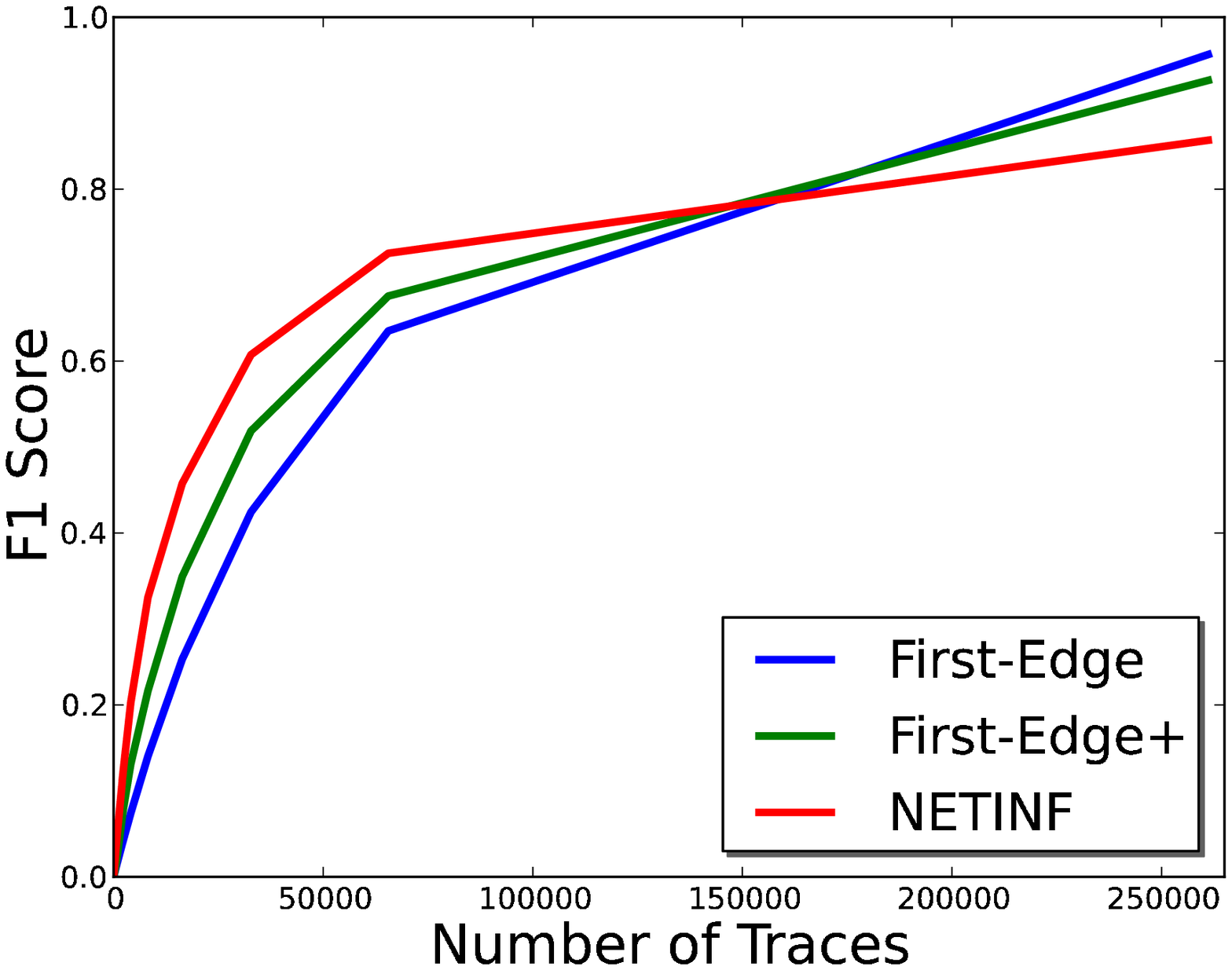}\label{fig:r4}}
    \end{tabular} 
    \end{center} 
    \vspace{-0.5cm}
    \caption{F1 score of the First-Edge, First-Edge$+$, and \NetInf algorithms applied to different real and synthetic networks against a varying number of traces. (best viewed in color)} 
    \label{fig:results}
\end{figure}

In First-Edge$+$, we use the memoryless property of the exponential distribution to establish the probability $p$ of an edge pertaining to a network $G$. The algorithm works as follows. Consider a node $u$ that appears as the root of a trace at time $t_0=0$. When $u$ spreads the epidemic, some node $v$ is going to be the next infected at time $t_1$, which was sampled from an exponential distribution with parameter $\lambda$. At time $t_1$, notice that there are exactly $d_u -1$ nodes waiting to be infected by $u$, and exactly $d_v-1$ waiting to be infected by  $v$, where $d_u$ and $d_v$ are the degrees of $u$ and $v$ respectively. At time $t_1$ any of these nodes is equally likely to be infected, due to the memoryless property. Moreover, the next node $w$ that appears in a trace after time $t_1$ is going to be infected by $u$ with probability $p_{(u,w)}=\frac{d_u-1}{d_u+d_v-2}$ and by $v$ with probability $p_{(v,w)}=\frac{d_v-1}{d_u+d_v-2}$. We can approximate\footnote{The exact probability depends on the number of edges between each of the nodes $u_1,\ldots,u_k$ and the rest of the graph.} this reasoning for larger prefixes of the trace: given a sequence  $u_1,\cdots, u_k$ of infected nodes starting at the source of the epidemic, the probability that $u_{k+1}$ is a neighbor of $u_i$ is roughly $p_{(u_i,u_{k+1})}\simeq\frac{d_{u_i}}{\sum_j d_{u_j}}$. Therefore, for every segment of a trace that starts at the source, we infer an edge $(u,v)$ if $p_{(u,v)}> p$, computed using the reconstructed degrees, where $p$ is a tunable parameter. In our experiments we arbitrarily chose $p=0.5$. 

Note that First-Edge+ may not terminate as soon as we have inferred enough edges, even in the event that all true positives have been found, an effect that degrades its precision performance. To prevent this, we keep a variable $T$, which can be thought of as the \emph{temperature} of the inference process. Let $M$ be a counter of the edges inferred at any given time during the inference process, and $\hat{E}$ be an estimate of the total number of edges, computed using the degree reconstruction algorithm in the pre-processing phase. We define $T=\frac{M}{\hat{E}}$ and run the algorithm as long as $T<1.0$. In addition,  whenever we infer a new edge, we flip a coin and remove, with probability $T$, a previously inferred edge with the lowest estimated probability of existence. Thus, while the network is ``cold'', i.e., many undiscovered edges, edges are rapidly added and a few are removed, which boosts the recall. When the network is ``warm'', i.e., the number of inferred edges approaches $|E|$, we carefully select edges by exchanging previously inferred ones with better choices, thereby  contributing to the precision.

Figure~\ref{fig:results} contrasts the performance of First-Edge, First-Edge$+$ and an existing network algorithm, \NetInf~\cite{Gomez-Rodriguez2010Inferring}, with respect to the F1 measure. \NetInf requires the number of edges in the network as input, and thus we give it an advantage, by setting the number of edges to the true cardinality of edges for each network. 

In Figures \ref{fig:r1} and \ref{fig:r2}, we observe that, as First-Edge$+$ and \NetInf are less conservative, their F1 performances have an advantage over First-Edge for small numbers of traces, with First-Edge+ approaching the performance to \NetInf. Interestingly, in Figure~\ref{fig:r3}, we see that First-Edge and First-Edge+ achieve perfect tree inference with roughly $5,000$ traces, which reflects a trace complexity in $\Omega(n)$ rather than in $O(\log n)$, which is the trace complexity of Algorithm~\ref{alg:tree}.\footnote{In our experiments Algorithm~\ref{alg:tree} consistently returned the true edge set without false positives with $O(\log n)$ traces for various networks of various sizes. Therefore, in the interest of space we omit the data from these experiments.} This result illustrates the relevance of the algorithms for special cases we developed in Section~\ref{sec:tail}. Last, we observe that $G_{n,p}$ random graphs seem to have very large trace complexity. This is shown in Figure~\ref{fig:r4}, where neither our algorithms nor \NetInf can achieve high inference performance, even for large numbers of traces.

In accordance with our discussion in Section~\ref{sec:firstedge-algo}, we confirm that, in practice, we need significantly fewer than $n*\Delta$ traces for inferring most of the edges. It is perhaps surprising that First-Edge+, which is extremely simple, achieves comparable performance to the more elaborate counterpart, \NetInf. In addition, while \NetInf reaches a plateau that limits its performance, First-Edge+ approaches perfect inference as the number of traces goes to $\Omega(n\Delta)$. In the cases in which \NetInf achieves higher performance than First-Edge+, the latter is never much worse than the former. This presents a practitioner with a trade-off between the two algorithms. For large networks, while First-Edge+ is extremely easy to implement and makes network inferences (in a preemptive fashion) in a matter of seconds, \NetInf takes a couple of hours to run to completion and requires the implementation of an elaborate algorithm.

\section{Conclusion}
\label{sec:conc}

Our goal is to provide the building blocks for a rigorous foundation to the rapidly-expanding network inference topic. Previous works have validated claims through experiments on relatively small graphs as compared to the large number of traces utilized, whereas the relation that binds these two quantities remains insufficiently understood. Accordingly, we believe that a solid foundation for the network inference problem remains a fundamental open question, and that works like \cite{Netrapalli12Learning}, as well as ours, provide the initial contributions toward that goal. 

Our results have direct applicability in the design of network inference algorithms. More specifically, we rigorously study how much useful information can be extracted from a trace for network inference, or more generally, the inference of network properties without reconstructing the network, such as the node degree distribution. We first show that, to perfectly reconstruct general graphs, nothing better than looking at the first pair of infected nodes in a trace can really be done. We additionally show that the remainder of a trace contains rich information that can reduce the trace complexity of the task for special case graphs. Finally, we build on the previous results to develop extremely simple and efficient reconstruction algorithms that exhibit competitive inference performance with the more elaborate and computationally costly ones.

Some open technical questions stemming from our work are immediately apparent. For instance, what is the true lower bound for perfect reconstruction? Is it $O(n^2)$, $O(n \Delta)$ or some other bound which, in the case of the clique, reduces to what we have shown? And, 
 are there other meaningful statistics apart from the degree distribution that can be efficiently recovered? For graphs with maximum degree $\Delta$, our perfect reconstruction algorithm has running time exponential in $\Delta$: is this exponential dependence necessary? And while the algorithm's trace complexity is polynomial in $\Delta$, the upper bound of $\tilde{O}(\Delta^9)$ proven here is far from matching the lower bound $\Omega(\Delta^{2-\epsilon})$; what is the correct dependence of trace complexity on $\Delta$?
The bounded-degree restriction, while natural, is unlikely to be satisfied by real-world networks; is perfect reconstruction possible for the types of graphs that are likely to occur in real-world situations?

Perhaps the most relevant  avenue for future research in our context is to go beyond the notion of perfect reconstruction. This notion, while quite reasonable as a first step, is not flexible enough to be deployed in practical situations. One would need to take into account the possibility of accepting some noise, i.e. some false positives as well as false negatives. So the main issue is to look for algorithms and lower bounds that are expressed as a function of the precision and recall one is willing to accept in an approximate reconstruction algorithm.

Finally, it would be very interesting to develop similar results, or perhaps even a theory, of trace complexity for other types of information spreading dynamics. 





\newpage

\appendix
\section{Proofs missing from Section~\ref{sec:lb}}\label{app:proofs}

We start by proving Lemma~\ref{lem:lb:comb}, which is the combinatorial heart of our lower bound.
\begin{proof}[Proof of Lemma~\ref{lem:lb:comb}]
Points 1,~2,~4 are either obvious or simple consequences of point 3, which we now consider.
For $a \ge 2$, let $P_{a,b}$ be the probability that $1,2$ appear (in any order) in the positions $a,b$ conditioned on the starting node being different from $1,2$. Moreover, let $P_{1,b}$ the probability that $1,2$ appear (in any order) in the positions $1,b$ conditioned on the starting node being one of $1,2$. Then $p_{a,b} = \frac{n-2}n P_{a,b}$, for $a \ge 2$, and $p_{1,b} = \frac2n P_{1,b}$.

We now compute $P_{a,b}$. First, we assume $a = 1$. Them
$$P_{1,b} =  \prod_{i=2}^{b-1} \frac{i \cdot (n-i-1)}{(i-1) \cdot (n-i) + (n-i-1)} \cdot  \frac{b-1}{(b-1)(n-b)+(n-b-1)}.$$

Now assume that $a \ge 2$. We have:
$$ P_{a,b} = \prod_{i=1}^{a-1} \frac{i(n-i-2)}{i(n-i)} \cdot \frac{a2}{a(n-a)} \cdot \prod_{i=a+1}^{b-1} \frac{i \cdot (n-i-1)}{(i-1) (n-i) + (n-i-1)} \cdot  \frac{b-1}{(b-1)(n-b)+(n-b-1)}.$$

By telescoping, specifically by 
$\prod_{i=s}^{t} \frac{n-i-2}{n-i} = \frac{(n-t-1)(n-t-2)}{(n-s)(n-s-1)}$, we can simplify the expression to:
$$ P_{a,b} = 2 \cdot \frac{n-a-1}{(n-1)(n-2)} \cdot \prod_{i=a+1}^{b-1} \frac{i \cdot (n-i-1)}{(i-1) (n-i) + (n-i-1)} \cdot  \frac{b-1}{(b-1)(n-b)+(n-b-1)}.$$
Moreover, by trying to simplify the product term, and by collecting the binomial, we get:
$$ P_{a,b} = \frac{n-a-1}{\binom {n-1}2} \cdot \prod_{i=a+1}^{b-1} \frac{1}{1 + \frac{i-1}{i  (n-i-1)}} \cdot  \frac{b-1}{(b-1)(n-b)+(n-b-1)} \quad a \ge 2$$

Then, observe that for each $a \ge 1$, $b > a$, (and, if $a = 1$, $b \ge 3$) we have:
$$p_{a,b} = \left(1 \pm O\left(\frac1n\right)\right) \cdot \frac{n-a-1}{\binom {n}2} \cdot \prod_{i=a+1}^{b-1} \frac{1}{1 + \frac{i-1}{i  (n-i-1)}} \cdot  \frac{b-1}{(b-1)(n-b)+(n-b-1)}$$

We highlight the product inside $p_{a,b}$'s expression:
\begin{align*}
\pi_{a,b} &= \prod_{i=a+1}^{b-1} \frac{1}{1 + \frac{i-1}{i  (n-i-1)}}  = \prod_{i=a+1}^{b-1} \frac1{1 - \frac{1}{i(n-i)}} \cdot \prod_{i=a+1}^{b-1} \left(\frac{1}{1 + \frac{i-1}{i  (n-i-1)}} \cdot \left(1 - \frac{1}{i(n-i)}\right) \right)\\
&=\prod_{i=a+1}^{b-1} \frac1{1 - \frac{1}{i(n-i)}} \cdot \prod_{i=a+1}^{b-1} \left(\frac{i(n-i-1)}{i(n-i)-1} \cdot \frac{i(n-i)-1}{i(n-i)} \right)\\
& = \prod_{i=a+1}^{b-1} \frac1{1 - \frac{1}{i(n-i)}} \cdot \prod_{i=a+1}^{b-1} \frac{n-i-1}{n-i}  = \prod_{i=a+1}^{b-1} \frac1{1 - \frac{1}{i(n-i)}} \cdot \prod_{i=a+1}^{b-1} \frac1{1+\frac1{n-i-1}}\\
&  =  \frac{n-b}{n-a-1} \cdot \prod_{i=a+1}^{b-1} \frac1{1 - \frac{1}{i(n-i)}}.
\end{align*}
We take the product of the denominators of the ratios, obtaining:
$$\prod_{i=a+1}^{b-1} \left(1 - \frac{1}{i(n-i)}\right) \ge \prod_{i=1}^{n-1} \left(1 - \frac1{i(n-i)}\right) \ge  1 - O\left(\frac{\ln n}n\right).$$
Therefore, we have
$$\pi_{a,b} = \left(1 + O\left(\frac{\ln n}n\right)\right) \frac{n-b}{n-a-1}.$$

We now turn back to $p_{a,b}$ expressions. Plugging in our approximation of $\pi_{a,b}$, we get:
$$p_{a,b} = \frac{1 + O\left(\frac{\ln n}n\right)}{\binom{n }2} \cdot \frac{b-1}{b-1+\frac{n-b-1}{n-b}}.$$
The term $\frac{b-1}{b-1+\frac{n-b-1}{n-b}}$ is bounded within $1 - \frac1b$ and $1$. Therefore, if $a \ge 2$, $b \ge a + 1$ (and if $a = 1$, $b \ge 3$), we have:
$$p_{a,b} = \frac{1 + O\left(\frac{\ln n}n\right) - O\left(\frac1b\right)}{\binom{n}2}.$$
\end{proof}

Before moving on the two main Lemmas, we state (a corollary of) the Berry-Esseen Theorem \cite{B41, E56} which will be a crucial part of their proofs.
\begin{theorem}[Berry-Esseen \cite{B41,E56}]\label{thm:be}
Let $Z_1,\ldots,Z_n$ be independent random variables, such that $E[Z_i] = 0$ for each $i=1,\ldots, n$, and such that $A = \sum_{i=1}^n E\left[\left|Z_i\right|^3\right]$ is finite. Then, if we let $B = \sqrt{\sum_{i=1}^n E[Z_i^2]}$ and $Z = \sum_{i=1}^n Z_i$, we have that
$$\Pr[Z > \delta \cdot B] \ge \frac12 - \Theta\left(\delta + \frac{A}{B^3}\right),$$
and
$$\Pr[Z < -\delta \cdot B] \ge \frac12 - \Theta\left(\delta + \frac{A}{B^3}\right).$$
\end{theorem}

We will now use our Lemma~\ref{lem:lb:comb}, and the Berry-Esseen Theorem (Theorem~\ref{thm:be}), to prove Lemma~\ref{lem:lb:perm}.

\begin{proof}[Proof of Lemma~\ref{lem:lb:perm}]
Let $\ell_{a,b}$ be the number of traces having one of the nodes in $\{1,2\}$ in position $a$, and the other in position $b$. Then, $\sum_{b=2}^n \sum_{a=1}^{b-1} \ell_{a,b} = \ell$.  We start by computing the likelihoods $\mathcal{L}_0, \mathcal{L}_1$ of $\mathcal{P}$ assuming that the unknown graph is, respectively, $G_0$ or $G_1$. We proved in Lemma~\ref{lem:lb:comb}, that the two likelihood of a trace only depends on the positions of $1$ and $2$. Therefore, if $p_{a,b}$ is the probability of obtaining a trace with $1,2$-positions equal to $a,b$ in the graph $G_1$, we have:
$$\frac{L_0(\mathcal{P})}{L_1(\mathcal{P})} = \frac{\binom{n}2^{-1}}{\prod_{a=1}^{n-1}\prod_{b=a+1}^{n} p_{a,b}^{\ell_{a,b}}} = \prod_{a=1}^{n-1}\prod_{b=a+1}^{n} (1+ d(a,b))^{-\ell_{a,b}}.$$

Regardless of the unknown graph, we have that the probability that there exists $b < \beta = \Theta(\sqrt{\log n})$ for which there exists at least one $a < b$ such that $\ell_{a,b} > 0$ is at most $O\left(\ell \cdot \frac{\beta^2}{n^2}\right) = o\left(\log^{-1} n\right)$. We condition on the opposite event. Then,
\begin{align*}
\mathcal{R} & = \ln \frac{L_0(\mathcal{P})}{L_1(\mathcal{P})} = \frac{\binom{n}2^{-1}}{\prod_{b=\beta}^{n} \prod_{a=1}^{b-1}  p_{a,b}^{\ell_{a,b}}} = - \sum_{b=\beta}^{n}\sum_{a=1}^{b-1} \left(\ell_{a,b} \ln (1+ d(a,b))\right) \\
&=  - \sum_{b=\beta}^{n}\sum_{a=1}^{b-1} \left(\ell_{a,b} \left(d(a,b) +O\left(d(a,b)^2\right)\right)\right),
\end{align*}
where  $d: \binom{[n]}2 \rightarrow \mathbf{R}$ be the function defined in the statement of Lemma~\ref{lem:lb:comb}.

We aim to prove that the random variable $\mathcal{R}$ is sufficiently anti-concentrated that, for any  unknown graph $G_i$, the probability that $\Pr[\mathcal{R} > 0 | G_i] > \frac12\pm o(1)$ and $\Pr[\mathcal{R} < 0 | G_i] > \frac12 \pm o(1)$. This will prove that one cannot guess what is the unknown graph with probability more than $\frac12 \pm o(1)$.

First, let $X_0$ and $X_1$ be two random variable having support $\binom{[n]}2$, the first uniform and the second distributed like the distribution $p$ of Lemma~\ref{lem:lb:comb}.

We will compute a number of expectations so to finally apply Berry-Esseen theorem. First, recall that $\sum_{b =2}^n \sum_{a=1}^{b-1} d(a,b) = 0$. Therefore, $$E[ d(X_0) ] = 0.$$

Moreover,
$$E[ d(X_1) ] = \sum_{b=2}^n \sum_{a=1}^{b-1}\left(p_{a,b} \cdot d(a,b)\right) = \sum_{b=2}^n \sum_{a=1}^{b-1}\left( \frac{1+ d(a,b)}{\binom{n}2} \cdot d(a,b)\right).$$
Recall that $\sum_{b=2}^{n} \sum_{a = 1}^{b-1} d(a,b) = 0$. Then
$$E[ d(X_1) ] = \binom{n}2^{-1} \cdot \sum_{b=2}^n \sum_{a=1}^{b-1} d(a,b)^2.$$
Therefore, $E[ d(X_1) ] \ge 0$. Moreover,
$$E[ d(X_1) ] \le O\left( \sum_{b=2}^{\Theta(n / \ln n)} \left(\frac b{n^2} \cdot \left(\frac1{b}\right)^2\right) + \sum_{b = \Theta(n / \ln n)}^n \left(\frac{b}{n^2} \cdot \left(\frac{\ln n}{n}\right)^2\right)\right) = O\left(\frac{\ln^2 n}{n^2}\right).$$

It follows that, for $i = 0,1$, we have
$$0 \le E[ d(X_i) ] \le O\left(\frac{\ln^2 n}{n^2}\right).$$

We now move to the second moments. First observe that, for $i=0,1$,
$$E[d(X_i)^2] = \Theta\left(\binom n2^{-1} \sum_{b=2}^n \sum_{a=1}^{b-1} p_{a,b}^2\right).$$

We lower bound both $E[ d(X_0)^2 ]$ and $E[ d(X_1)^2 ]$ with
$$E[d(X_i)^2] = \Omega\left( \sum_{b=\Theta(n/\ln n)}^n \left(\frac{b}{n^2} \cdot \left(\frac{\ln n}{n}\right)^2\right)\right) =  \Omega\left(\frac{\ln^2 n}{n^2}\right).$$
Analogously, we have that $E[ d(X_0)^2 ]$ and $E[ d(X_1)^2 ]$ can both be upper bounded by
$$O\left(\sum_{b=2}^{\Theta(n / \ln n)} \left(\frac{b}{n^2} \cdot \left(\frac{1}{b}\right)^2\right) + \sum_{b=\Theta(n/\ln n)}^n \left(\frac{b}{n^2} \cdot \left(\frac{\ln n}{n}\right)^2\right)\right) =  O\left(\frac{\ln^2 n}{n^2}\right).$$

We then have, for $i = 0,1$, 
$$E[ d(X_i)^2 ] = \Theta\left(\frac{\ln^2 n}{n^2}\right).$$
Moreover, the variance of $d(X_i)$, $i=0,1$, is equal to $S^2 = \Theta\left(\frac{\ln^2 n}{n^2}\right)$.

By linearity of expectation, regardless of the unknown graph, if we let $C = \Theta\left(\ell \cdot \frac{\ln^2 n}{n^2}\right)$, we have that
$$-C\le E[\mathcal{R}] \le C.$$

We  upper bound both $E[ |d(X_0)|^3]$ and $E[ |d(X_1)|^3 ]$ with
$$O\left(\sum_{b=2}^{\Theta(n / \ln n)} \left(\frac b{n^2} \cdot \left(\frac1{b}\right)^3\right) + \sum_{b = \Theta(n / \ln n)}^n \left(\frac{b}{n^2} \cdot \left(\frac{\ln n}{n}\right)^3\right)\right) \le O\left(\frac{1}{n^2} + \frac{\ln^3 n}{n^3}\right) = O\left(\frac1{n^2}\right).$$

It follows that
\begin{align*}
K &= E[|d(X_i) - E[d(X_i)]|^3] \le E[ \max(8|d(X_i)|^3, 8|E[d(X_i)]|^3)) ] \\
&\le O\left( \max\left( E[|d(X_i)|^3], \frac{\ln^6 n}{n^6}\right)\right) \le O\left(\frac1{n^2}\right).
\end{align*}

Now, we apply the Berry-Esseen bound with $A \le \ell  K = o\left(\frac1{\ln^2 n}\right)$ and $B = \Theta(\sqrt{\ell S^2}) = \Theta(1)$.
 We compute the error term of the Berry-Esseen theorem:
$$\frac A{B^3} \le  O\left(\frac1{\ln^2 n}\right) = o(1).$$
Therefore, $\mathcal{R}$ will behave approximately like a gaussian in a radius of (at least) $\omega(1)$ standard deviations $B$ around its mean. Observe that the standard deviation $B$ satisfies $B = \Theta\left(\sqrt{C}\right)$. Since $C = o(1)$ (by $\ell = o\left(\frac{n}{\ln^2n}\right)$), we have $B = \omega(C)$. Therefore, regardless of the unknown graph, the probability that $\mathcal{R}$ will be positive is $\frac12 \pm o(1)$.\end{proof}

We finally prove Lemma~\ref{lem:lb:waitingtimes}, which deals with the likelihoods of the waiting times in the traces.

\begin{proof}[Proof of Lemma~\ref{lem:lb:waitingtimes}]
Let $f_0(x) = t e^{-t x}$ and $f_1(x)= (t-1) e^{-(t-1) x}$ be two exponential density functions with parameters $t$ and $t-1$.
Since we will be considering ratio of likelihoods, we compute the ratio of the two densities:
$$\frac{f_0(x)}{f_1(x)} = \left(1 + \frac1{t-1}\right) \cdot e^{-x}.$$ 

Observe that, if $q = o(t)$, it holds that
\begin{equation}\label{eqn:ratio_approx}
\left(\frac{f_0(\frac{1}{t})}{f_1(\frac{1}{t})}\right)^{q} = 1 + \frac{q}{2t^2} + o\left(\frac q{t^2}\right).
\end{equation}

Let $\delta_{i,j} = 1$ if, in the $j$th trace, exactly one of the nodes in $\{1,2\}$ was within the first $i$ informed nodes; otherwise, let $\delta_{i,j} = 0$. 
Then $\ell_i = \sum_{j=1}^{\ell} \delta_{i,j}$.

\smallskip

For $i=1,\ldots,n-1$ and $j=1,\ldots,\ell_i$, let $t_{i,j}$ be the  time waited (from the last time a node was informed) to inform the $(i+1)$th node in the $j$th of the traces having exactly one of the two nodes $1,2$ in the first $i$ positions. By the memoryless property of the exponential random variables, and by the fact that the minimum of $n$ iid $\expv(\lambda)$ random variables is distributed like $\expv(n\lambda)$, we have that (once we condition on the $\ell_i$'s) the $t_{i,j}$'s variables are independent, and that $t_{i,j}$ is distributed like $\expv(c \lambda)$ where $c$ is the size of the cut induced by the first $i$ nodes of the $j$th trace (of those having one of the nodes $1,2$ within the first $i$ nodes).
Further, from the scaling property of the exponential random variables, we have that $\lambda t_{i,j}$ is distributed like $\expv(c)$.

\smallskip

Let $T = \lambda \cdot \sum_{i=1}^{n-1}\sum_{j=1}^{\ell_i} t_{i,j} \delta_{i,j}$.
Let $T_0$ be the random variable $T$ conditioned on $G_0$, and let $T_1$ be the random variable $T$ conditioned on $G_1$ (observe that, since $T$ is conditioned on $\ell_1,\ldots,\ell_{n-1}$, both $T_0$ and $T_1$ will also be conditioned on $\ell_1,\ldots,\ell_{n-1}$). Then,
$$ T_0 = \sum_{\substack{i,j\\\delta_{i,j}=1}} (\lambda t_{i,j}) = \sum_{\substack{i,j\\\delta_{i,j}=1}} \expv( i \cdot (n-i)) = \sum_{i=1}^{n-1} \sum_{j=1}^{\ell_i} \expv(i \cdot (n-i)),$$
$$ T_1 = \sum_{\substack{i,j\\\delta_{i,j}=1}} (\lambda t_{i,j}) = \sum_{\substack{i,j\\\delta_{i,j}=1}} \expv(i \cdot (n-i)-1) = \sum_{i=1}^{n-1} \sum_{j=1}^{\ell_i} \expv(i \cdot (n-i) - 1) .$$

Now, let $X$ be distributed like $\expv(x)$, for some $x > 0$.
In general, we have $E[X^k] = \frac{k!}{x^{k}}$.
If we let $Y = X - E[X]$, we obtain $E[Y] = 0$. Moreover, we have
$$E[Y^2] = E[X^2] - 2 E[X] E[X] + E[X]^2 = E[X^2] - E[X]^2 = \frac{2}{x^2} - \frac{1}{x^2} = x^{-2.}$$

We also have,
\begin{align*}
E[|Y|^3] & = E\left[\left|X - E[X]\right|^3\right]  \le E\left[\left(X + E[X]\right)^3\right] = E\left[X^3 + 3X^2 E[X] + 3X E[X]^2 + E[X]^3\right] \\
& = E[X^3] + 3 E[X^2] E[X] + 4 E[X]^3 = 16x^{-3},
\end{align*}
where the inequality follows from $X \ge 0$.

Let us consider $ T_k$, $k = 0,1$. They are the sum, over $i = 1,\ldots,n-1$, of $\ell_i$ independent exponential random variables with parameters $i \cdot(n-i) - k$. If we center each of those variables in $0$ (that is, we remove the expected value of each exponential variable), obtaining --- say --- variables $Y_{k,i,j} = \expv(i \cdot (n-i) - k) - \frac1{i \cdot (n-i) - k}$, we can write $ T_k$ as:
$$ T_k = \sum_{i=1}^{n-1} \frac{\ell_i}{i \cdot (n-i) - k} + \sum_{i=1}^{n-1} \sum_{j=1}^{\ell_i} Y_{k,i,j} = E[ T_k] + \sum_{i=1}^{n-1} \sum_{j=1}^{\ell_i} Y_{k,i,j}.$$

Let us bound the absolute difference between the expected values of $T_0$ and $T_1$, recalling that  $\ell_i = \Theta(\alpha i (n-i))$, for each $i=1,\ldots,n-1$:
$$E[T_1] - E[T_0] = \sum_{i=1}^{n-1} \left(\ell_i \cdot \left(\frac1{i(n-i)-1} - \frac1{i(n-i)}\right)\right) =\sum_{i=1}^{n-1} \frac{\ell_i}{i^2(n-i)^2-i(n-i)} = \Theta\left(\alpha \cdot \frac{\log n}{n}\right) = D.$$

We now aim to show that $\Pr[T_k > E[T_0]] = \frac12 \pm o(1)$, for both $k=0$ and $k=1$. To do so, we use the Berry-Esseen theorem (Theorem~\ref{thm:be}). We apply it to our collection of variables $Y_{k,i,j}$ (which will play the $Z_i$ variables' role in the Berry-Esseen theorem). We get:
$$A \le O \left( \sum_{i=1}^{n-1} \frac{\ell_i}{i^3(n-i)^3} \right),$$
and
$$B = O\left(\sqrt{\sum_{i=1}^{n-1} \frac{\ell_i}{i^2(n-i)^2}}\right).$$

 Therefore,
$$A \le O\left(\alpha \cdot \sum_{i=1}^{n-1} \frac1{i^2(n-i)^2}\right) = O\left(\alpha \cdot n^{-2}\right),$$
and
$$B  = \Theta\left(\sqrt{\alpha \cdot \sum_{i=1}^{n-1} \frac1{i(n-i)}}\right) = \Theta\left( \sqrt{ \alpha \cdot \frac{\log n}{n}}\right).$$

Therefore, $\frac{A}{B^3} \le O\left(\left(\alpha n \log^3 n\right)^{-\frac12}\right)$.
We set $\delta = \frac1{\sqrt{\log n}}$.
Then, for each $k,k' \in \{0,1\}$, the probability that $T_k$ exceeds $E[T_{k'}]$ by at least $\Omega\left(\sqrt{\frac{\alpha \log n}n}\right)$ is at least $\frac12 - \Theta\left(\frac1{\sqrt{\log n}}\right) = \frac12 - o(1)$.

\medskip

Now consider, the likelihoods $L_0$, $L_1$ of the sequence of waiting times with each of the two graphs. 
We have:
$$L_0(\mathcal{W}) = \prod_{i=1}^{n-1} \prod_{j=1}^{\ell} ( \lambda \cdot i \cdot (n-i) \cdot e^{-\lambda i(n-i)  t_{i,j}})$$
$$L_1(\mathcal{W}) = \prod_{i=1}^{n-1} \prod_{j=1}^{\ell} ( \lambda \cdot (i \cdot (n-i) - \delta_{i,j}) \cdot e^{-\lambda (i(n-i) - \delta_{i,j}) t_{i,j}})$$

Then,
\begin{align*}
\frac{L_0(\mathcal{W})}{L_1(\mathcal{W})} &= \prod_{i=1}^{n-1} \prod_{j=1}^{\ell_i} \left(\left(1 + \frac1{i\cdot (n-i)-1}\right) \cdot e^{-t_{i,j}}\right)\\
& = \prod_{i=1}^{n-1} \prod_{j=1}^{\ell_i} \left(\left(1 + \frac1{i\cdot (n-i)-1}\right) \cdot e^{-\frac 1{i(n-i)}}\right) \cdot \prod_{i=1}^{n-1} \prod_{j=1}^{\ell_i} e^{-t_{i,j} + \frac 1{i(n-i)}}\\
& =\prod_{i=1}^{n-1} \left(\left(1 + \frac1{i\cdot (n-i)-1}\right)^{\ell_i} \cdot e^{-\frac {\ell_i}{i(n-i)}}\right) \cdot \prod_{i=1}^{n-1} \prod_{j=1}^{\ell_i} e^{-t_{i,j} + \frac 1{i(n-i)}}
\end{align*}
Since $\ell_i = o\left(i \cdot (n-i)\right)$, we can apply Equation~(\ref{eqn:ratio_approx}), and the former product equals
\begin{align*}
\frac{L_0(\mathcal{W})}{L_1(\mathcal{W})} & = \prod_{i=1}^{n-1} \left(1 + \Theta\left(\frac{\ell_i}{i^2\cdot (n-i)^2}\right)\right)  \cdot \prod_{i=1}^{n-1} \prod_{j=1}^{\ell_i} e^{-t_{i,j} + \frac 1{i(n-i)}}\\
& = \prod_{i=1}^{n-1} \left(1 + \Theta\left(\frac{\alpha}{i\cdot (n-i)}\right)\right)  \cdot \prod_{i=1}^{n-1} \prod_{j=1}^{\ell_i} e^{-t_{i,j} + \frac 1{i(n-i)}}
\end{align*}
The former product simplifies to $1 \pm \Theta\left(\alpha \cdot \frac{\log n}n\right) = 1 \pm o(1)$. Therefore,
\begin{align*}
\frac{L_0(\mathcal{W})}{L_1(\mathcal{W})}& = (1 \pm o(1))  \cdot e^{\sum_{i=1}^{n-1} \frac{\ell_i}{i(n-i)} - T}\\
& = (1 \pm o(1))  \cdot e^{E[T_0] - T}.
\end{align*}

Therefore, $L_0(\mathcal{W}) > L_1(\mathcal{W})$ if $T \le E[T_0] - 1$, and $L_0(\mathcal{W}) < L_1(\mathcal{W})$ if $T \ge E[T_0] + 1$. We have that $ | E[T_0 - T_1] | \le D = \Theta\left(\alpha \cdot \frac{\log n}{n}\right)$; moreover, the probability that the difference between $T$ and its expectation is at least $2D$, and the probability that it is at most $-2D$, are both $\frac12 - o(1)$, since the standard deviation $B$ satisfies $B = \omega(D)$.

Therefore, 
the probability that the likelihood of graph $G_0$ is higher than the likelihood of graph $G_1$ is $ \frac12 \pm o(1)$, regardless of whether the unknown graph was $G_0$ or $G_1$. The proof is concluded.
\end{proof}

\newpage
\bibliographystyle{abbrv}
\bibliography{paper} 

\end{document}